\documentclass[12pt,centertags,reqno]{amsart}

\usepackage[foot]{amsaddr}
\usepackage{latexsym}
\usepackage[english]{babel}
\usepackage[T1]{fontenc}
\usepackage[numbers]{natbib}
\usepackage{amssymb}
\usepackage{fancyhdr}
\usepackage{url}
\usepackage{hyperref}
\usepackage{verbatim}
\usepackage{leftidx}
\usepackage{bm}
\usepackage{color,graphicx}
\usepackage{mathrsfs, mathtools}
\usepackage{stmaryrd}

\usepackage{marvosym}
\mathtoolsset{showonlyrefs}

\numberwithin{equation}{section} \makeatletter
\renewcommand{\subsection}{\@startsection
{subsection}{2}{0mm}{\baselineskip}{-0.25cm}
{\normalfont\normalsize\bf}} \makeatother

\newtheorem{theorem}{Theorem}[section]
\newtheorem{lemma}[theorem]{Lemma}

\newtheorem{definition}[theorem]{Definition}
\newtheorem{remark}[theorem]{Remark}
\newtheorem{proposition}[theorem]{Proposition}
\newtheorem{example}[theorem]{Example}
\newtheorem{assumption}[theorem]{Assumption}

\newtheorem{problem}[theorem]{Problem}

\def \A {\mathcal A}

\def \E {\mathcal E}
\def \F {\mathcal F}
\def \G {\mathcal G}

\def \L {\mathcal L}

\def \P {\mathbf P}
\def \Q {\mathbf Q}
\def \I {{\mathbf 1}}

\def \R {\mathbb R}
\def \bF {\mathbb F}
\def \bG {\mathbb G}

\newcommand{\ud}{\mathrm d}
\newcommand{\ds}{\displaystyle}
\newcommand{\esp}[2][\mathbb E] {#1\left[#2\right]}
\newcommand{\espq}[2][\mathbb E^\Q] {#1\left[#2\right]}

\newcommand{\condespftilde}[2][\widetilde \F_t]       {\mathbb E\left.\left[#2\right|#1\right]}

\begin{document}

\author[C.~Ceci]{Claudia  Ceci}
\address{Claudia  Ceci, Department of Economics,
University ``G. D'Annunzio'' of Chieti-Pescara, Viale Pindaro, 42,
65127 Pescara, Italy.}

\author[K.~Colaneri]{Katia Colaneri}
\address{Katia Colaneri, Department of Economics and Finance,
 University of Rome Tor Vergata, Via Columbia, 2, 00133 Roma,  Italy.}

\author[A.~Cretarola]{Alessandra Cretarola}
\address{Alessandra Cretarola ({\large \Letter}), Department of Mathematics and Computer Science,
 University of Perugia, Via Luigi Vanvitelli, 1, 06123 Perugia, Italy.  {\em ORCID identifier:} \textit{https://orcid.org/0000-0003-1324-9342}. {\em E-mail address}: \href{mailto:alessandra.cretarola@unipg.it}{alessandra.cretarola@unipg.it}.}

\title[Indifference pricing via BSDEs under partial info]{Indifference pricing of pure endowments via BSDEs under partial information}

\date{}

\maketitle

\begin{abstract}
In this paper we investigate the pricing problem of a
pure endowment contract when the insurance company has a limited information on the mortality intensity of the policyholder. The payoff of this kind of policies depends on the residual life time of the insured as well as the trend of a portfolio traded in the financial market, where
investments in a riskless asset, a risky asset and a longevity bond are allowed.
We propose a modeling framework that takes into account mutual dependence between the financial and the insurance markets via an
observable stochastic process, which affects the risky asset and the mortality index dynamics.
Since the market is incomplete due to the presence of basis risk,
in alternative to arbitrage pricing we use expected utility maximization under exponential preferences as evaluation approach, which leads to the so-called indifference price. Under partial information this methodology requires filtering techniques that can reduce the original control problem to an equivalent problem in complete information.
 Using stochastic dynamics techniques, we characterize the indifference price of the insurance derivative in terms of the solutions of  two backward stochastic differential equations.  Finally, we discuss two special cases where we get a more explicit representation of the indifference price process.
\end{abstract}





\section{Introduction}

This paper deals with the pricing problem of some particular unit linked life insurance contracts, which are long term insurance policies between a policyholder and an insurance company. These kinds of contract are hybrid financial products embodying banking, securities and insurance components. Indeed,
the payoff depends on the insured remaining lifetime (insurance risk) and on the performance of the underlying stock or portfolio (financial risk). In this paper we focus on a {\em pure endowment} policy, which promises to pay an agreed amount if the policyholder is still alive on a specified future date.

In these contracts benefits are random. This makes for instance, traditional valuation principles for pricing life insurance products with deterministic payoffs, inappropriate. Since the 70's, it was clear that theory of financial valuation, suitably combined with mortality, was the right way forward, see, e.g. \citet{brennan1976pricing, boyle1977equilibrium, aase1994pricing}. In these papers the Black \& Scholes pricing methodology is applied under the hypotheses of market completeness and independence between financial and insurance setting. Since then, many efforts have been done to relax the assumption of completeness and several approaches have been proposed, for instance in  \citet{moller2003indifference, ludkovski2008indifference, bayraktar2009valuation, delong2012no, blanchet2017model}. However the problem of incorporating some kind of dependence between the financial and the insurance market, which is empirically observed, has started to be addressed only recently.

The goal of this paper is to study the pricing problem of a pure endowment life insurance contract in a general modeling framework that accounts for {\em mutual dependence} between the financial and the insurance markets and {\em partial information} of the insurance company on the mortality intensity of the policyholder. Precisely, we consider a discounted financial market with a riskless asset, a risky asset and a longevity bond written on the mortality index of the same age cohort of population of the policyholder.
We assume that the dynamics of the risky asset and the mortality index are governed by general diffusion processes with coefficients that depend on the same {\em observable} stochastic process representing economic and environmental factors.
The insurance company issues a pure endowment policy with maturity of $T$ years for an individual whose remaining lifetime is represented by a random time.
The partial information scenario refers to the situation where the insurance company knows at any time if the policyholder is still alive but cannot directly observe her/his mortality intensity, which is influenced by an exogenous unobservable factor describing the social/health status of the individual.

The insurance contract can be treated as a contingent claim in the hybrid market model given by the financial securities and the insurance portfolio, and the goal is to study the pricing problem for the insurance company. In our setting mortality intensity of the population and of the policyholder do not coincide in general. This translates into the presence of a basis risk that, even in the context of complete information, does not permit perfect replication of the contract via self-financing strategies and the financial-insurance market is incomplete.
In other words,  the insurance company cannot perfectly hedge its exposure by investing in a hedging instrument (the longevity bond), which
is based on the mortality of the whole population, rather then that of the insured, leaving a residual amount of risk, see e.g. \citet{biagini2016risk} for a deeper discussion on this issue. Therefore, in alternative to arbitrage pricing we use expected utility maximization under exponential preferences as evaluation approach, which leads to the so-called {\em indifference price} (see e.g. \citet{henderson2009} for a survey).
This is a well-known technique for pricing derivatives in incomplete markets and has been successfully applied to valuate insurance derivatives under full information  in e.g. \citet{becherer2003rational, ludkovski2008indifference, delong2009indifference, delong2013backward, eichler2017utility, liang2017indifference}. However, to the best of our knowledge, applications of this methodology  to insurance derivatives under partial information is an open problem.
We formulate the pricing problem from the perspective of the insurance company that issues the pure endowment and study the resulting control problem via a backward stochastic differential equation (in short BSDE) approach.
Our main result is a characterization of the indifference price process in terms of unique solutions of BSDEs arising from two different optimization problems that involve the insurance claim and a pure investment. The former is a stochastic control problem under partial information and presents a jump due to mortality of the policyholder. Our setting is non standard since mortality intensity is stochastic and unobservable. By applying filtering techniques, we transform the optimization problem into an equivalent problem in complete information involving only observable processes. Then,
using the structure of the jump component, we can transform the equation into a continuous BSDE with quadratic-exponential driver. Although existence and uniqueness of the solution to such type of BSDEs have been studied very recently in the literature and they only cover a setting under full information. Therefore, these results do not apply straightforwardly to our framework and
need to be suitably adjusted. To the best of our knowledge, this is the first time that the problem of evaluating  an insurance claim via indifference pricing when there are restrictions on the available information in a general setting with mutual dependence between the financial and the insurance framework, is investigated.

We conclude with a few remarks on our filtering results (see Appendix \ref{appendix:filtering}).
Due to partial information on the policyholder mortality intensity, we need to apply filtering theory to solve the optimization problem with the insurance derivative.
We stress that we deal with a non trivial filtering problem that combines partial information of an unobservable general Markov process with progressive filtration enlargement, due to the presence of a (random) time of death.
To estimate the mortality intensity of the policyholder, we use observations of the mortality intensity of the reference population enlarged with the information on the time of death. In addition, we obtain an explicit expression for the filter via linearization with respect to the non-enlarged filtration.
We also discuss uniqueness of the solution of the filtering equation using an approach based on the Filtered Martingale Problem. The resulting filter is an infinite dimensional probability measure valued process, that shows up as state variable in the optimization and therefore, we study the problem following the BSDE approach.

\indent The outline of the paper is as follows. Section \ref{sec:model}  introduces the mathematical framework and describes
the combined model
allowing for a mutual dependence between the financial and the insurance markets and a limited information on the mortality intensity of the policyholders. The pricing problem formulation under partial information via utility indifference pricing can be found in Section \ref{sec:pricing_problem}. In Section \ref{sec:bsdeGEN} we study the resulting stochastic control problems following a BSDE approach and characterize the log-value process corresponding to the problem with (without, respectively) the pure endowment contract in terms of the solution to a quadratic-exponential (quadratic, respectively) BSDE. A characterization of the indifference price of the pure endowment policy is given in Section \ref{sec:indiff-price}. Some concluding remarks can be found in Section \ref{sec:conclusion}.
We address the filtering problem in Appendix \ref{appendix:filtering}. How to compute the longevity bond price process is shown in Appendix \ref{app:longevity}.
A Technical result can be found in Appendix \ref{appendix:tech_res}.

\section{Modeling framework} \label{sec:model}

We consider the problem of an insurance company that issues a unit-linked life insurance contract. This type of contract has a relevant link with the financial market. Indeed, the value of the policy is determined by the performance of the underlying stock or portfolio. Moreover, it also depends on the remaining lifetime of the policyholder. Therefore, we construct a {\em combined financial-insurance market model} and treat the life insurance policy as a contingent claim. We will
define the suitable modeling framework via the progressive enlargement of filtration approach, which allows for possible dependence between the financial market and the insurance portfolio.

We start by fixing a complete probability space  $(\Omega,\F,\P)$ endowed with a complete and right continuous filtration $\bF=\{\F_t, \ t \in[ 0 , T]\}$, where $T>0$ is a fixed and finite time horizon, such that  $\F=\F_T$, $\F_0=\{\Omega, \emptyset\}$.

On this filtered probability space we consider a Markov stochastic process $Z=\{Z_t,\ t \in [0,T]\}$
 with c\'adl\'ag trajectories and values in a locally compact and separable space $\mathcal{Z}$. We assume that $Z$ is {\em not be observable} by the insurance company and denote by $\bF^Z=\{\F^{Z}_t, \ t \in [0,T]\}$, with $\F^{Z}_t:=\sigma\{Z_u, \ 0 \leq u \leq t\}$, for each $t \in [0,T]$, its natural filtration.
We may interpret the process $Z$ as an environmental process describing the social level/ health status of an individual to be insured.

We assume that the probability space supports three $\P$-independent standard $\bF$-Brownian motions $W^j=\{W_t^j,\ t \in [0,T]\}$, with $W_0^j=0$, for each $j=1,2,3$, which
are also $\P$-independent of the stochastic factor $Z$. Here, $W^j$, for $j=1,2,3$, are supposed to drive the underlying financial market (see Subsection \ref{sec:market}) and the mortality intensity defined on the same age cohort of the population, see \eqref{def:index_intensity}.
Now, denote by $\bF^{W^j}=\{\F^{W^j}_t, \ t \in [0,T]\}$, with $\F^{W^j}_t:=\sigma\{W_u^j, \ 0 \leq u \leq t\}$,  $j=1,2,3$, for every $t \in [0,T]$, the canonical filtrations of $W^j$, for every $j=1,2,3$, respectively. In addition, set
\begin{equation}\label{ftilde}
\widetilde \F_t:=\F_t^{W^1} \vee \F_t^{W^2} \vee \F_t^{W^3}, \quad t \in [0,T]
\end{equation}
and $\widetilde \bF=\{\widetilde \F_t,\ t \in [0,T]\}$.
We assume that the reference filtration $\bF$ is given by
\begin{equation}\label{eq:filtrF}
\bF = \widetilde \bF \vee \bF^{Z},
\end{equation}
completed by $\P$-null sets, so that, it contains all knowledge of the financial-insurance market except for the information regarding the policyholder survival time.

\subsection{Construction of the death time and mortality intensities}\label{sec:tau}

We consider an individual aged $l$ at time $0$ to be insured. Let $\mu=\{\mu_t,\ t \in [0,T]\}$ be an $\bF$-adapted process modeling the mortality intensity of an equivalent age cohort of the population. This process is observable and can be computed using publicly available data of the survivor index $S^\mu=\{S_t^\mu,\ t \in [0,T]\}$ given by
\begin{equation}
S_t^\mu:=\exp\left(-\int_0^t \mu_s \ud s\right), \quad t \in [0,T].
\end{equation}
We assume that $\mu$ evolves according to the following stochastic differential equation:
\begin{equation}\label{def:index_intensity}
\ud \mu_t =b^\mu(t,\mu_t,Y_t) \ud t + \sigma^\mu(t,\mu_t,Y_t) \ud W_t^2,\quad \mu_0 \in \R^+,
\end{equation}
where $Y=\{Y_t,\ t \in [0,T]\}$ is an {\em observable} stochastic process representing economic and environmental factors, satisfying
\begin{equation}\label{def:Y}
\ud Y_t =b^Y(t,Y_t) \ud t + \sigma^Y(t,Y_t) \ud W_t^3,\quad Y_0=y_0 \in \R.
\end{equation}
Here, functions $b^\mu(t,\mu,y)$, $b^Y(t,y)$, $\sigma^\mu(t,\mu,y)>0$, and $\sigma^Y(t,y)>0$ are  measurable and such that the system of equations \eqref{def:index_intensity}-\eqref{def:Y} admits a unique strong solution, with $ \mu_t \geq 0$, $\P$-a.s., for all $t \in [0,T]$, see e.g. \citet{oksendal2013}.

We remark that it is not difficult to generalize the results of the paper to the case where Brownian motions $W^2$ and $W^3$ are correlated, and our choice for independence is just due to notational reasons.

We now introduce a few key examples to better illustrate our model.
\begin{example}
We consider a generalized Cox-Ingersoll-Ross model to represent the trend of the mortality intensity $\mu$, see , e.g.  \citet{dahl2004stochastic,biffis2005}. In  \citet{biffis2005},  mortality intensity of the sample population follows an affine dynamics with stochastic drift, given by
\begin{equation}\label{CIR}
\ud \mu_t=a^\mu\left(Y_t-\mu_t\right) \ud t + \sigma^\mu \sqrt{\mu_t} \ud W_t^2, \quad \mu_0\in \R^+,
\end{equation}
where $a^\mu,\ \sigma^\mu \in \R^+$  and the process $Y$ satisfies
\begin{equation}\label{CIR2}
\ud Y_t=a^Y(b^Y(t)-Y_t) \ud t + \sigma^Y \sqrt{Y_t-b^*(t)}\ud W_t^3, \quad Y_0=y_0\in \R^+,
\end{equation}
for some nonnegative, bounded and continuous functions $b^Y(t), b^*(t)>0$ and $\sigma^Y>0$.
It is known that this model describes well mortality intensity and it is quite flexible to capture stylized features, such as fluctuations around a target mean (given here by $b^Y(t)$), and $\P$-a.s. positivity of the intensity process $\mu$ which is satisfied for instance when $b^Y(t)\geq b^*(t)$ for every $t \in [0,T]$  and $y_0 \ge b^*(0)$. See \citet{biffis2005} for a deeper discussion.  
\end{example}

\begin{example}\label{ex1}
A simplified setting is obtained taking the coefficients in the dynamics of the process $\mu$ independent of the stochastic factor $Y$, that is
\begin{equation}\label{def:index_intensity_indep}
\ud \mu_t =b^\mu(t,\mu_t) \ud t + \sigma^\mu(t,\mu_t) \ud W_t^2,\quad \mu_0 \in \R^+.
\end{equation}
This describes a situation where the mortality intensity is not affected by the observable factor, and it will be used as a benchmark in the sequel.
\end{example}

\begin{example} \label{C1}
For comparison reasons, we also consider an example which will be discussed in more details later. Assume that the stochastic factor $Y$ follows the dynamics given by equation \eqref{def:Y} and that the mortality intensity $\mu$ is described by a nonnegative function $\mu(t,y)$, so that for all $t \in [0,T]$ we have that $\mu_t=\mu(t, Y_t)$. If $\mu \in C^{1,2}([0,T]\times\R)$,
by applying It\^{o}'s formula we have that
\begin{equation}\label{eq:ex2}
\ud \mu_t=\left(\frac{\partial \mu}{\partial t}(t, Y_t)+\frac{\partial \mu}{\partial y}(t, Y_t)b^Y(t, Y_t)+\frac{1}{2}\frac{\partial^2 \mu}{\partial y^2}(t, Y_t)\right)\ud t + \frac{\partial \mu}{\partial y}(t, Y_t)\sigma^Y(t, Y_t) \ud W^3_t,
\end{equation}
with $\mu_0=\mu(0, Y_0)$. The setting provided by equations \eqref{def:Y} and \eqref{eq:ex2} corresponds to the case where the mortality intensity $\mu$ and the observable stochastic factor $Y$ are perfectly correlated.
\end{example}

To describe the stochastic residual lifetime of the individual, we adopt the canonical construction of a random time in terms of a given hazard process, in analogy to  reduced-form credit risk models. For this reason,
we shall postulate that the underlying filtered probability space $(\Omega, \F, \bF, \P)$ is sufficiently rich to support a random variable $\Theta$ having unit exponential distribution, $\P$-independent of $\F_T$. Let $\lambda:[0,T]\times \R^+\times \mathcal Z \longrightarrow (0, + \infty)$ be a positive function such that $\esp{\int_0^T \lambda(s,\mu_s,Z_s)\ud s} < \infty$
and define the random time $\tau:\Omega\longrightarrow \R^+$ by setting
\begin{equation}
\tau:=\inf\left\{t \geq 0 : \ \int_0^t \lambda(s,\mu_s,Z_s)\ud s \geq \Theta\right\}.
\end{equation}
In this framework, $\tau$ represents the remaining lifetime of an individual and $\lambda$ is the $\bF$-{\em mortality intensity} process. The associated $\bF$-hazard process is given by  $\left\{\int_0^t \lambda(s,\mu_s,Z_s) \ud s, \ t \in [0,T]\right\}$.
Note that by the $\P$-independence assumption on $\Theta$ and the $\F_t$-measurability of $\int_0^t \lambda(s,\mu_s,Z_s) \ud s$, we get that
\begin{equation}\label{eq:cond_distribution_F}
\P(\tau> t|\F_t)=\P\left(\int_0^t\lambda(s,\mu_s,Z_s) \ud s < \Theta\Big{|}\F_t\right)=e^{-\int_0^t\lambda(s,\mu_s,Z_s)\ud s}, \quad t\in[ 0, T],
\end{equation}
 and the following property of the canonical construction of the remaining lifetime $\tau$ holds
 \begin{equation}\label{eq:cond_equivalente_hpH}
\P(\tau\leq t|\F_t)= \P(\tau\leq t|\F_T), \quad  t \in [0, T],
\end{equation}
see, for instance, \citet[Section 8.2.1]{bielecki2002}.

\begin{remark}
It is intuitively clear that, in general the mortality rate of the insured $\lambda$ is different from that of its age cohort, $\mu$. In our model $\lambda$ is a function of $\mu$ as well as the unobservable process $Z$.
 A possible choice could be
$$
\lambda(t,\mu_t,Z_t) = \mu_t \tilde \lambda(Z_t), \quad t \in [0,T],$$
where $\tilde \lambda(Z_t)$ is a strictly positive function of the environmental factor $Z_t$ for every $t$, meaning that when $\tilde \lambda(z)<1$ the risk of the policyholder is smaller than that of the reference population, and bigger if $\tilde \lambda(z)>1$.
\end{remark}
Since the random time $\tau$ is not a stopping time with respect to filtration $\bF$,  we introduce an enlarged filtration that makes $\tau$ a stopping time. First, we define the death indicator process $H = \{H_t, \ t \in [0,T]\}$ associated to $\tau$ as follows
\begin{equation}
H_t:=\I_{\{\tau\leq t\}}, \quad  t \in[0,T],
\end{equation}
and set $\F^H_t:=\sigma\{H_u, \ 0\leq u \leq t\}$, for every $t \in [0,T]$.
Let $\bG=\{\G_t,\ t\in [0, T]\}$ be the enlarged filtration given by
\begin{equation}\label{eq:filtrG}
\G_t:=\F_t\vee\F_t^H, \quad t \in [0,T].
\end{equation}
Then, $\bG$ is the smallest filtration which contains $\bF$ and such that $\tau$ is a $\bG$-stopping time. Filtration $\bG$ plays the role of the market full information: it contains all the knowledge about the insurance and the market portfolio.

As an immediate consequence of the canonical construction of the residual lifetime $\tau$, we get that the so-called {\em martingale invariance property} between filtrations $\bF$ and $\bG$ holds, i.e. every $(\bF, \P)$-(local) martingale is also a $(\bG, \P)$-(local) martingale, see \citet{bremaud1978changes}. Moreover, the process $\{H_t-\int_0^{t\wedge \tau} \lambda(s,\mu_s,Z_s)\ud s,\ t \in [0,T]\}$ is a $(\bG,\P)$-martingale and $\tau$ is a totally inaccessible $\bG$-stopping time.

\subsection{The combined financial-insurance market model} \label{sec:market}
We define a combined financial-insurance market on the filtered probability space $(\Omega, \G, \bG, \P)$, with $\G=\G_T$, where the tradeable securities are given by a riskless asset, a risky asset and a longevity bond. We assume that the price process of a risk free asset is equal to $1$ at any time,
and that the risky asset has discounted price process $S^1=\{S_t^1,\ t\in [0, T]\}$
 given by the following geometric diffusion with coefficients affected by the economic and environmental factor $Y$
 \begin{equation} \label{def:S}
 \ud S_t^1 = S_t^1\left(\mu^S(t, Y_t)\ud t +\sigma^S(t,Y_t) \ud W^1_t\right), \quad S_0^1 = s_0^1 \in \R^+.
 \end{equation}
The longevity bond has discounted price process $S^2=\{S^2_t,\ t\in [0, T]\}$ satisfying the following stochastic differential equation with coefficients depending on the equivalent age cohort mortality intensity $\mu$ and the stochastic factor $Y$
\begin{align}\label{eq:long_bond}
\ud S^2_t= S^2_t\left(\mu^B(t, \mu_t, Y_t)\ud t +c^B(t,\mu_t,Y_t)\ud W_t^2 + d^B(t,\mu_t,Y_t)\ud W_t^3\right), \quad S_0^2=s_0^2\in \R^+.
\end{align}
Here $\mu^S(t,y)$, $\sigma^S(t,y)>0$, $\mu^B(t,\mu,y)$, $c^B(t,\mu,y)>0$ and  $d^B(t,\mu,y)>0$ are measurable functions such that the system of equations \eqref{def:index_intensity}-\eqref{def:Y}-\eqref{def:S}-\eqref{eq:long_bond} admits a unique strong solution. The motivation for the dynamics of the longevity bond price process is given in Appendix \ref{app:longevity}.

It is clear that the proposed modeling framework allows for mutual dependence between the financial and the insurance markets via the stochastic factor $Y$, which affects all stochastic processes dynamics given in \eqref{def:index_intensity}, \eqref{def:S} and
\eqref{eq:long_bond}.

Notice that the primary financial-insurance market defined on the probability space $(\Omega, \widetilde\F_T, \widetilde{\bF}, \P)$ given by the riskless asset, the risky asset and the longevity bond is incomplete. In the following, we discuss two examples where the primary financial-insurance market is instead complete.

\begin{example}\label{ex11}
We consider the case where the primary market is independent of the observable factor $Y$. This means that the mortality intensity $\mu$ is described as in  Example \ref{ex1} and the dynamics of the risky asset and the longevity bond are given, respectively, by

\begin{align}
 \label{def:S_indep}
 \ud S_t^1 = S_t^1\left(\mu^S(t)\ud t +\sigma^S(t) \ud W^1_t\right), \quad S_0^1 = s_0^1 \in \R^+,\\
\label{eq:long_bond_indep}
\ud S^2_t= S^2_t\left(\mu^B(t, \mu_t)\ud t +c^B(t,\mu_t)\ud W_t^2 \right), \quad S_0^2=s_0^2\in \R^+.
\end{align}
It is easy to see that in this case the primary market is complete.
\end{example}

\begin{example}\label{C11}
Now, we consider a market model where the mortality intensity $\mu$ is given as in Example \ref{C1}. In this setting the price dynamics of the risky asset is as in equation \eqref{def:S} and the price of the longevity bond follows
\begin{align}\label{eq:long_bond1}
\ud S^2_t= S^2_t\left(\mu^B(t, Y_t)\ud t + d^B(t,Y_t)\ud W_t^3\right), \quad S_0^2=s_0^2\in \R^+.
\end{align}
We notice that, although we still have mutual dependence between the risky asset and the longevity bond price dynamics, the primary financial-insurance market is complete.
\end{example}

We suppose that the insurance company issues a unit-linked life insurance policy. This is a long term insurance contract between the policyholder and insurance company whose benefits are linked to financial assets. In particular, we consider a pure endowment contract with maturity of $T$ years, which can be defined as follows.
\begin{definition}\label{def:pure endowment}
A {\em pure endowment} contract with maturity $T$ is a life insurance policy where the sum insured is paid at time $T$ if the insured is still alive. The associated final value is given by the random variable
\begin{equation}\label{eq:pure endowment}
 G_T := \xi \I_{\{\tau>T\}},
\end{equation}
where $\xi \in L^2(\widetilde \F_T,\P)$ represents the payoff of a European-type contingent claim with maturity $T$.
\end{definition}

From now on we make the following boundedness assumption on the payoff of the insurance policy.
\begin{assumption} \label{ass:V-bound}
The random variable $\xi$ in \eqref{eq:pure endowment} is bounded, that is,
\begin{equation}\label{bound}
|\xi| \leq k \quad  \P-a.s. \end{equation}
with $k$ positive constant.
\end{assumption}

Contracts of this type are for instance unit-linked policies with capped benefits, whose payoff is given by $\xi=\min\{S_T^1, K_T\}$ for some bounded $K_T>0$. In this contract the insurance company aims to cap the benefit amount in order to limit the risk exposure to a large increase of the price of the underlying. Another example is given by policies with capped benefits and a guarantee where benefits payable at time $T$ are specified by $\xi=\min\{\max\{S_T^1, G_T\}, K_T\}$ with $K_T > G_T >0$. In this case the insured has both limited downside risk and upside potential and the insurance company is also exposed to a limited financial risk.

Throughout the rest of the paper, we work under the following integrability conditions.
\begin{assumption} \label{mgloc}
We assume that the function $\lambda(t,\mu,z)>0$ is bounded and continuous  w.r.t. $z \in \mathcal Z$ and that
\begin{align}
&\int_0^T \left\{\mu^S(t, Y_t)^2+\sigma^S(t,Y_t)^2\right\} \ud t  < \infty, \quad \P-\mbox{a.s.},\\
& \int_0^T\left\{\mu^B(t, \mu_t, Y_t)^2 + c^B(t, \mu_t, Y_t)^2 + d^B(t, \mu_t, Y_t)^2 \right \}\ud t < \infty, \quad \P-\mbox{a.s}.,\\
&\int_0^T \left\{\left(\frac{\mu^S(t, Y_t)}{\sigma^S(t,Y_t)}\right)^2 + \frac{\mu^B(t,\mu_t,Y_t)^2}{c^B(t,\mu_t,Y_t)^2 + d^B(t, \mu_t, Y_t)^2}\right\}
 \ud t < \infty, \quad \P-\mbox{a.s.}.
\end{align}
\end{assumption}
Note that the discounted asset price processes $S^1$ and $S^2$ are continuous $(\bF,\P)$-semimartingales and also $(\bG,\P)$-semimartingales. As a consequence the underlying financial-insurance market model is arbitrage-free.

\subsection{Available information and filtering}

We assume that the insurance company observes prices of the assets  negotiated on the markets,  $S^1$ and  $S^2$ and the death time of the insured $\tau$, but it has not full information about the intensity  mortality intensity $\lambda$, which depends on $Z$.
Therefore, the available information to the insurance company is given by $\widetilde\bG =\{\widetilde\G_t,\ t\in [0, T]\}$ where
\begin{equation}\label{gtilde}
\widetilde \G_t:=\widetilde \F_t\vee\F_t^H, \quad t \in [0,T],
\end{equation}
where $\widetilde \bF$ is defined in \eqref{ftilde}.
Note that
\begin{equation*}\widetilde \bG\subseteq \bG =  \bF\vee\bF^H = \widetilde \bF \vee  \bF^Z \vee\bF^H, \end{equation*}
and  we refer to $\widetilde \bG$ as the available information to the insurance company. We assume throughout the paper that all filtrations satisfy the usual hypotheses of completeness and right-continuity. \\

The intensity of the mortality process $H$ with respect the information flow can be characterized via a filtering approach. Denote by  $\L^Z$ the Markov generator of $Z$ and by  $\mathcal{D}\subseteq C_b(\mathcal Z)$ the domain of the generator, that is for every function $f \in \mathcal{D} \subseteq C_b(\mathcal Z)$
\begin{align}\label{eq:f_semimg}f(Z_t) = f(z_0) + \int_0^t {\L}^Z f(Z_s) \ud s + M^Z_t, \quad t \in [0,T],
\end{align}
for some $(\bF^Z, \P)$-martingale $M^Z = \{M^Z_t,\ t \in [0,T]\}$, with $z_0\in \mathcal Z$.\\

In the rest of the paper we assume that the following conditions holds.
\begin{assumption}\label{mgp}
\begin{itemize}
\item[]
\item[(i)]  The martingale problem for the operator $\L^Z$, for any initial value $z_0 \in \mathcal Z$, is well posed on the space of  c\'adl\'ag trajectories with values in $\mathcal Z$, $D_{\mathcal Z}[0, T]$;
\item[(ii)] ${\L}^Z f \in C_b(\mathcal Z)$ for any  $f \in \mathcal{ D} $;
\item[(iii)]  $\mathcal{D}$ is an algebra dense in $C_b(\mathcal Z)$.
\end{itemize}
\end{assumption}

We define the filter process by setting
\begin{equation}
\pi_t(f) := \esp{ f(Z_t)  \Big{|}  \widetilde\G_t}, \quad t \in [0,T],\label{eq:filter}
\end{equation}
for every measurable function $f$ such that $\esp{|f(Z_t)|}< \infty$, for each $t \in [0,T]$. It is known that $\pi(f)$ is a probability measure-valued process with  c\'{a}dl\'{a}g  trajectories (see \citet{kurtz1988unique}), which provides the conditional law of $Z$ given the information flow $\widetilde\bG$.
Then, the  $\widetilde \bG$-predictable intensity of  $H$  is given by
\begin{equation}\label{fint}
(1-H_{t^-}) \pi_{t^-}(\lambda), \quad t \in [0,T],
\end{equation}
where $\pi_{t^-}$ denotes the left version of $\pi_{t}$ and $\pi_{t}(\lambda)$ is short for $\pi_t(\lambda(t, \mu_t, \cdot))$. This means that the compensated process $M^\tau=\{M^\tau_t, \ t \in [0,T]\}$ defined as
\begin{equation}\label{def:mtau}
M^\tau_t:=H_t-\int_0^{t \wedge\tau} \pi_{s^-}(\lambda)\ud s=H_t-\int_0^{t} (1-H_{s^-}) \pi_{s^-}(\lambda)\ud s,\quad  t \in [0,T],
\end{equation}
turns out to be a $(\widetilde \bG, \P)$-martingale.

It is well-known that the exact computation of the filter can be obtained only in a few cases (for instance, the linear Gaussian case) and that in general, numerical approximations must be employed to solve the filtering equation. In our setting, however, we can explicitly characterize the filter by giving the following representation (the proof  is discussed in detail
in Appendix \ref{appendix:filtering}). 

\begin{proposition}\label{nuova_filter}

The filter $\pi = \{\pi_t,\ t \in [0,T]\}$ coincides on $\{t < T \wedge \tau\}$ with  the $\widetilde \bF$-adapted process $ \widehat \pi=\{\widehat \pi_t,\ t \in [0,T]\}$
\begin{equation}\label{eq:hat_pi}
\widehat  \pi_t(f)(\omega)  = \frac{\esp{ f(Z_t) e^{-\int_0^t  \lambda (s,\mu_s(\omega), Z_s)  \ud s} }} {\esp{ e^{-\int_0^t  \lambda (s,\mu_s(\omega), Z_s)  \ud s} }}.
\end{equation}
Moreover for $t =\tau<T$,
 $$\pi_{\tau}(f)=  \frac{\widehat  \pi_{\tau^-}(\lambda f)}{ \widehat  \pi_{\tau^-}(\lambda)},$$
and for $ \tau <t\leq T$,
$$ \pi_t(f) :=  {\mathbb E}_{\tau, \pi_\tau}[ f(Z_t) ],$$
where $ {\mathbb E}_{\tau, \pi_\tau}$ denotes the conditional expectation given the law of $Z$ at time $\tau$ equals to $\pi_\tau$.
\end{proposition}


\section{The pricing problem} \label{sec:pricing_problem}

The goal of this paper is to price a pure endowment policy with payoff given by \eqref{eq:pure endowment} in a partially observable setting, where the insurance company does not have access to the full
information given by the filtration $\bG$. In particular, the insurance company is not allowed to observe the
evolution of the stochastic factor $Z$, which therefore implies that its decisions are based on the observation filtration $\widetilde \bG$. Moreover, we recall that our general setting accounts for possible mutual dependence between the financial and the insurance framework, which is a desirable characteristic when dealing with mortality derivatives. Indeed, nowadays it is commonly recognized that, in the long term, demographic changes may affect the economy and vice-versa. Unit-linked life insurance contracts have previously been studied under partial information in \citet{ceci2015hedging, ceci2017hedging}, where the goal was to solve the hedging problem in an incomplete market via local risk-minimization. Precisely, in  \citet{ceci2015hedging} the independence between financial market and  insurance portfolio was assumed, while in \citet{ceci2017hedging} the authors considered a more general situation where mutual dependence between  financial and insurance context is allowed.

Insurance-financial market models are typically incomplete due to the fact that mortality events are in general not hedgeable. This implies that insurance contracts may have different risk-neutral prices. One of the criteria that can be used to compute a fair price corresponds to identify insurance company's preferences towards the risk via a specific utility function and maximize the expected utility whether it holds the insurance claim or not. In other words, to characterize the utility indifference price, which is the price $p$ that makes the insurance company indifferent between not selling the policy and selling the policy at price $p$ now and paying the benefits at maturity.

In this paper, we follow a utility indifference pricing approach assuming that the insurance company is endowed with an exponential utility function of the form
\begin{equation}
U(x)= -e^{-\alpha x}, \quad x \in \R,
\end{equation}
where $\alpha >0$ is a given constant, representing a coefficient of absolute risk aversion. This form of the utility function is frequently assumed in the literature and allows for more explicit computations of the price.

Suppose that the insurance company has initial wealth $x$ and it invests this amount in the money market account, in the risky asset and in the longevity bond, following a self-financing strategy.

\noindent Set ${\mathbf{S}}=(S^1,S^2)^\top$  and let ${\boldsymbol{\theta}} = (\theta^1, \theta^2)^\top=\{(\theta^1_t, \theta^2_t)^\top, \ t \in [0,T]\}$  be the amount of wealth invested in the risky asset and the longevity bond respectively. Given an initial wealth $x_0 \in \R^+$ the portfolio value $X^\theta=\{X_t^\theta,\ t \in [0,T]\}$ satisfies
\begin{align}
\ud X^\theta_t = & {\boldsymbol{\theta}}^\top_t \frac{\ud {\mathbf{S}}_t}{{\mathbf{S}}_t}
=\left(\theta^1_t \mu^S(t, Y_t)+ \theta^2_t \mu^B(t, \mu_t, Y_t)  \right) \ud t \\
& \qquad \qquad + \theta^1_t \sigma^S(t, Y_t) \ud W_t^1 +  \theta^2_t \left[ c^B(t, \mu_t, Y_t) \ud W_t^2 +
d^B(t, \mu_t, Y_t) \ud W_t^3 \right], \label{def:port_value}
 \end{align}
 with  $X^\theta_0 = x_0 \in \R^+$.

In the case where it sells the insurance contract
the information at its disposal is given by the filtration $\widetilde \bG$ defined in \eqref{gtilde},  whereas in the case of pure investment by the filtration $\widetilde \bF$ given in \eqref{ftilde}.
The set of admissible strategies is defined below.
\begin{definition}\label{def:selffinancing}
 An admissible strategy is a self-financing portfolio identified by an $\widetilde \bF$-predictable (or even $\widetilde \bG$-predictable), $\R^2$-valued process ${\boldsymbol{\theta}} = (\theta^1, \theta^2)^\top$ such that
\begin{align}
&\int_0^T \left\{\left( \theta^1_t \sigma^S(t, Y_t)\right)^2 +|\theta^1_t\mu^S(t, Y_t)|\right\} \ud t< \infty, \quad \P-\mbox{a.s.},\label{eq:integ_1}\\
&\int_0^T\left\{ (\theta^2_t)^2 \left[c^B(t, \mu_t, Y_t)^2 + d^B(t, \mu_t, Y_t)^2 \right] + |\theta^2_t\mu^B(t,\mu_t,Y_t)|\right\}\ud t< \infty, \quad \P-\mbox{a.s.} \label{eq:integ_2}
\end{align}
 and that the family
\begin{equation}\label{eq:integ_3}
\{  e^{-\alpha X^\theta_\eta}, \mbox{ for all } \widetilde\bF-\mbox{stopping times} \ \eta\leq T\} \ (\mbox{ or } \widetilde\bG-\mbox{stopping times})
\end{equation}
is $\P$-uniformly integrable.
We denote by $\mathcal A(\widetilde \bF)$ and $\mathcal A(\widetilde \bG)$, the set of admissible $\widetilde \bF$-predictable and $\widetilde \bG$-predictable strategies, respectively.
\end{definition}

\noindent In order to characterize the indifference price, we introduce the optimal investment problems with and without the insurance derivative.
First, suppose that, at time $t$, the insurance company sells a pure endowment contract with payoff given by \eqref{eq:pure endowment}. Then, the  goal of the insurance company is the following.
\begin{problem}\label{p1}
To maximize the expected utility of its terminal wealth, i.e. to solve
\begin{equation}\label{eq:problem}
\sup_{{\boldsymbol{\theta}} \in \mathcal A(\widetilde \bG)}\esp{-e^{-\alpha \left(X^\theta_T-G_T\right)}}.
\end{equation}
\end{problem}
For every $(t,x) \in [0,T] \times \R^+$, the value process in a dynamic framework is given by
\begin{equation}\label{eq:dyn_problem}
\widetilde V_t(x):= {{\rm ess}\inf}_{{\boldsymbol{\theta}} \in \mathcal A_t(\widetilde \bG)}  \esp{e^{- \alpha \left(x + \int_t^T {\boldsymbol{\theta}}^\top_u \frac{\ud {\mathbf{S}}_u}{ {\mathbf{S}}_u} - G_T\right)} \Big{|} \widetilde \G_t} = e^{- \alpha x} V_t^G,
\end{equation}
where $\mathcal A_t(\widetilde \bG)$  denotes the class of admissible $\widetilde \bG$-predictable controls on the interval $[t,T]$ and the process $V^G=\{V_t^G,\ t \in [0,T]\}$ is given by
\begin{equation}\label{VG}
V_t^G := {{\rm ess}\inf}_{{\boldsymbol{\theta}} \in \mathcal A_t(\widetilde \bG)}  \esp{e^{- \alpha \left(\int_t^T {\boldsymbol{\theta}}^\top_u \frac{\ud {\mathbf{S}}_u}{{\mathbf{S}}_u} - G_T\right)} \Big{|} \widetilde \G_t}, \quad t \in [0,T].
\end{equation}
Hence, Problem
\ref{p1} can be written in terms of the process $V^G$
$$
\sup_{{\boldsymbol{\theta}} \in \mathcal A(\widetilde \bG)}\esp{-e^{-\alpha \left(X^\theta_T-G_T\right)}} = - e^{- \alpha x_0} V_0^G.
$$

\begin{remark}\label{r0}
Since ${\boldsymbol{\theta}}=(\theta^1,\theta^2)^\top=(0,0)^\top$ is an admissible strategy we get that for every $t \in [0,T]$,
 $V_t^G \leq \esp{e^{\alpha G_T} | \widetilde \G_t }$  $\P$-a.s.  which implies  by Assumption \ref{ass:V-bound}  that  $V^G_t \leq  e ^{\alpha k}$, $ \P-a.s.$ for each $t \in [0,T]$. Clearly, $V^G_t \geq 0$ $\P$-a.s. and if there exists an optimal strategy, then  $V_t^G >0$  $\P$-a.s..
\end{remark}

\begin{remark}
We remark here that, since $G_T = \xi \I_{\{\tau>T\}}$  and  the mortality intensity $\lambda(t,\mu_{t},Z_{t^-})(1-H_{t^-})$ is not observable by the insurance company, we are dealing with a utility maximization problem in a partial information framework.  The idea is to consider an  equivalent control problem under full information where the  unobservable intensity of $\tau$ is replaced by its filtered estimate, see \eqref{fint}.
\end{remark}
Now, we consider the case where the insurance company simply invests its wealth in the market, without writing the insurance derivative. Then, the objective is the following.
\begin{problem}\label{p2}
To maximize the expected utility of his/her terminal wealth, i.e. to solve
\begin{equation}\label{eq:problem2}
\sup_{{\boldsymbol{\theta}} \in \mathcal A(\widetilde \bF)}\esp{-e^{-\alpha X_T^\theta}}.
\end{equation}
\end{problem}
For every $(t,x) \in [0,T]\times \R^+$, the associated value process is given by
\begin{equation}\label{eq:problem0}
\widetilde V^0_t(x):= {{\rm ess}\inf}_{{\boldsymbol{\theta}} \in \mathcal A_t(\widetilde \bF)}  \esp{e^{- \alpha\left(x + \int_t^T {\boldsymbol{\theta}}^\top_u \frac{\ud {\mathbf{S}}_u}{ {\mathbf{S}}_u}\right) } \Big{|} \widetilde \F_t} =  e^{- \alpha x} V^0_t,
\end{equation}
where  $\mathcal A_t(\widetilde \bF)$ is  the set  of admissible  $\widetilde \bF$-predictable controls on the interval $[t,T]$ and $V^0=\{V_t^0,\ t \in [0,T]\}$ is defined as
\begin{equation}\label{V0}
V^0_t := {{\rm ess}\inf}_{{\boldsymbol{\theta}} \in \mathcal A_t(\widetilde \bF)}  \esp{e^{- \alpha \int_t^T {\boldsymbol{\theta}}^\top_u \frac{\ud {\mathbf{S}}_u}{{\mathbf{S}}_u} } \Big{|} \widetilde \F_t}  \quad t \in [0,T].
\end{equation}

\begin{definition}

The {\em utility indifference price} or {\em  reservation price}  $p^\alpha$ of the insurance company related to the pure endowment contract
is defined at any time $t \in [0,T]$ as  the $\widetilde \bG$-adapted process  implicit solution to the equation
$$
\widetilde V_t(x + p_t^\alpha) = \widetilde V^0_t(x).
$$
This means that starting at time $t$ with capital  $x$, the insurance company has the same maximal utility from selling the insurance product for $p_t^\alpha$  at time $t$ and solely trading on $(t, T ]$ without  writing the contract.
\end{definition}

\noindent If $V_t^G>0$ and $V^0_t>0$, $\P$-a.s., for every $t \in [0,T]$,  we get that $p^\alpha$ does not depend on the initial capital $x$ and it is given by
\begin{equation}\label{ip}
p_t^\alpha:= \frac{1}{\alpha} \log \left (\frac{V_t^G}{V^0_t}  \right),  \quad t \in [0,T].
\end{equation}
Note that, since $G_T=0$ if $\tau<T$, then
\begin{equation}\label{eq:vg}V^G_t  \I_{\{\tau \leq  t\}}= V^0_t\I_{\{\tau \leq  t\}}  \quad t \in [0,T].
\end{equation}
Therefore, the indifference price $p^\alpha$ is given by
\begin{equation}\label{eq:indiff-price}
p_t^\alpha= \frac{1}{\alpha} \left(\log (V_t^G)-\log(V^0_t)\right)\I_{\{\tau > t\}}  \quad t \in [0,T],
\end{equation}
provided that $V^G_t$, $V^0_t>0$  on $\{\tau>t\}$.

\section{Optimization problems via BSDEs} \label{sec:bsdeGEN}

The goal of this section is to characterize dynamically the value processes $V^G$ and $V^0$ given in \eqref{VG} and \eqref{V0} respectively, and corresponding to the stochastic control problems with and without the insurance derivative, by using a BSDE-based approach.
The BSDE method works well in non-Markovian settings where the  classical stochastic control approach based on the Hamilton-Jacobi-Bellman equation does not apply.  Several papers (see, e.g. \citet{el1997backward,ceci2011utility, lim2015portfolio} and references therein) deal with stochastic optimization problems in Finance by means of BSDEs. Moreover, this approach is also well suited to solve stochastic control problems under partial information in presence of an infinite-dimensional filter process, which is the situation in our paper, see e.g. \citet{ceci2012utility, ceci2010optimal, papanicolaou2019}, where partially observed power utility maximization problems are solved by applying this approach.

First, we define some spaces that are used throughout the sequel:
\begin{itemize}
\item $L^2(W;\widetilde \bG)$ (respectively $L_{loc}^2(W;\widetilde \bG)$)
is the set of $\R$-valued $\widetilde \bG$-predictable processes $u=\{u_t, \ t \in [0,T]\}$ such that
\begin{equation}\label{int_cond}
\esp{\int_0^T |u_s|^2\ud s} < \infty \quad \left(\mbox{respectively}\ \int_0^T |u_s|^2\ud s < \infty \ \P-\mbox{a.s.}\right).
\end{equation}
Moreover, $L^2(W;\widetilde \bF)$ (respectively $L_{loc}^2(W;\widetilde \bF)$) is the set  of $\R$-valued $\widetilde \bF$-predictable processes $u=\{u_t, \ t \in [0,T]\}$ satisfying \eqref{int_cond}.
\item $L^p(M^\tau)$ (respectively $L_{loc}^p(M^\tau)$) for $p=1,2$ is the set of all $\R$-valued $\widetilde \bG$-predictable processes $\eta=\{\eta_t, \ t \in [0,T]\}$ such that
\begin{gather}
\esp{\int_0^T |\eta_s|^p(1-H_s)\pi_s(\lambda)\ud s} < \infty\\ \left(\mbox{respectively}\ \int_0^T |\eta_s|^p(1-H_s)\pi_s(\lambda)\ud s < \infty,\ \P-\mbox{a.s.}\right).
\end{gather}
\end{itemize}

\subsection{The problem with life insurance liabilities} \label{sec:bsde}

We first aim to characterize the value process $V^G$ in \eqref{VG}.

The following proposition is a general verification result.
\begin{proposition}\label{VER}
If there exists  a $\widetilde \bG$-adapted process $D=\{D_t,\ t \in [0,T]\}$ such that
\begin{itemize}
\item[(i)] $D_T = \alpha G_T$;
\item[(ii)] $\{e ^{-\alpha X^\theta_t + D_t},\ t \in [0,T]\}$  is a $(\widetilde \bG,\P)$-submartingale for any ${\boldsymbol{\theta}} \in \A(\widetilde \bG)$

then, $V_t^G \geq  e^{D_t}$, for every $t \in [0,T]$, $\P$-a.s..
\end{itemize}
 If in addition
 \begin{itemize}
\item [(iii)] $\{e ^{-\alpha X^{{\boldsymbol{\theta}}^*}_t + D_t},\ t \in [0,T]\}$  is a $(\widetilde \bG,\P)$-martingale  for some ${\boldsymbol{\theta}}^* \in \A(\widetilde \bG)$
\end{itemize}

then, $V_t^G = e^{D_t}$, for every $t \in [0,T]$, $\P$-a.s. and ${\boldsymbol{\theta}}^*$ is an optimal investment strategy for Problem \ref{p1}.
\end{proposition}

\begin{proof}
Let $D$ be a $\widetilde \bG$-adapted process satisfying conditions (i) and (ii). Then
$$
\esp{e ^{-\alpha (X^\theta_T - G_T)} \Big{|} \widetilde \G_t } = \esp{e ^{-\alpha X^\theta_T + D_T}  \Big{|} \widetilde \G_t}
\geq e^{-\alpha X^\theta_t + D_t},
$$
for any ${\boldsymbol{\theta}} \in \A(\widetilde \bG)$, which implies
$
\esp{e^{- \alpha \left(\int_t^T {\boldsymbol{\theta}}^\top_u \frac{\ud {\mathbf{S}}_u}{{\mathbf{S}}_u} - G_T\right)} \Big{|} \widetilde \G_t }  \geq  e^{D_t}$, hence  $V_t^G \geq  e^{D_t}$ for every $t \in [0,T]$, $\P$-a.s..

Moreover, if  $\{e ^{-\alpha X^{{\boldsymbol{\theta}}^*}_t + D_t},\ t \in [0,T]\}$  is a $(\widetilde \bG,\P)$-martingale  for some ${\boldsymbol{\theta}}^* \in \A(\widetilde \bG)$, one obtains 

$$
\esp{e^{- \alpha \left(\int_t^T {\boldsymbol{\theta}}^\top_u \frac{\ud {\mathbf{S}}_u}{{\mathbf{S}}_u} - G_T\right)} \Big{|} \widetilde \G_t }  \geq  e^{D_t} = \esp{e^{- \alpha \left(\int_t^T {\boldsymbol{\theta}_u^*}^\top \frac{\ud {\mathbf{S}}_u}{{\mathbf{S}}_u} - G_T\right)} \Big{|} \widetilde \G_t }  , \quad  {\boldsymbol{\theta}} \in \A(\widetilde \bG).
$$
This implies
$$
{{\rm ess}\inf}_{{\boldsymbol{\theta}} \in \mathcal A_t(\widetilde \bG)}  \esp{e^{- \alpha \left(\int_t^T {\boldsymbol{\theta}}^\top_u \frac{\ud {\mathbf{S}}_u}{{\mathbf{S}}_u} - G_T\right)} \Big{|} \widetilde \G_t } = e^{D_t} = \esp{e^{- \alpha \left(\int_t^T {{\boldsymbol{\theta}}_u^*}^\top \frac{\ud {\mathbf{S}}_u}{{\mathbf{S}}_u} - G_T\right)} \Big{|} \widetilde \G_t},
$$
which concludes the proof.
\end{proof}

Next, we characterize the optimal strategy and the log-value process $\log V^G$ via the solution of a BSDE with quadratic-exponential driver.

\begin{theorem}\label{verification}
Let $(U^G, \gamma^1, \gamma^2, \gamma^3, \gamma^4)$, with $U^G$ bounded, $\gamma^i \in L^2(W;\widetilde \bG)$, $i=1,2,3$,
$\gamma^4 \in L^1(M^{\tau})$, be a solution to the BSDE
\begin{align}
&U^G_t = \alpha G_T - \int_t^T \sum_{i=1}^3\gamma^i_s \ud W^i_s - \int_t^T   \gamma^4_s  \ud M^\tau_s - \int_t^T  \widetilde f(s, \gamma^1_s,  \gamma^2_s,  \gamma^3_s, \gamma^4_s ) \ud s ,\label{eq:BSDEq-e}
\end{align}
where
\begin{align}
&\widetilde f(t,\gamma^1, \gamma^2,  \gamma^3,  \gamma^4 )= - (e^{ \gamma^4} - \gamma^4 - 1)(1 - H_t) \pi_{t}(\lambda)   - \frac{1}{2} ((\gamma^1)^2+( \gamma^2)^2+(\gamma^3)^2) \\
&\quad + \frac{1}{2} \left(\frac{\mu^S(t, Y_t)}{\sigma^S(t, Y_t)}+ \gamma^1\right)^2+ \frac{1}{2} \frac{\left(\mu^B(t, \mu_t, Y_t)+c^B(t,\mu_t,Y_t)\gamma^2+d^B(t,\mu_t,Y_t) \gamma^3\right)^2}{c^B(t,\mu_t,Y_t)^2+d^B(t,\mu_t,Y_t)^2}.
\end{align}
Then, $V^G_t \geq e^{U^G_t}$, for every $t \in [0,T]$, $\P$-a.s. and
$$
\left\{\E \left( \int_0^. (e^{\gamma^4_s} -1)  \ud M^\tau_s  + \int_0^. \sum_{i=1}^3\gamma^i_s \ud W^i_s \right)_t,\ t \in [0,T]\right\}
$$
is a bounded $(\widetilde \bG,\P)$-martingale. The notation $\E$ stands for the stochastic exponential.
Moreover, if ${\boldsymbol{\theta}}^*=(\theta_t^{1,*},\theta_t^{2,*})^\top \in \A(\widetilde \bG)$, with
\begin{align}
\theta_t^{1,*} &=  \frac {\mu^S(t, Y_t) } { \alpha \sigma^S(t,Y_t)^2}+\frac { \gamma_t^1 } { \alpha \sigma^S(t,Y_t)}, \quad t \in [0,T], \label{theta1*}\\
\theta_t^{2,*}& = \frac { \mu^B(t, \mu_t,Y_t) +c^B(t, \mu_t,Y_t) \gamma^2_t + d^B(t, \mu_t,Y_t)  \gamma^3_t} { \alpha[c^B(t, \mu_t,Y_t)^2 + d^B(t, \mu_t,Y_t)^2]}, \quad t \in [0,T],  \label{theta2*}
\end{align}
then,  ${\boldsymbol{\theta}}^*$ is an optimal strategy in the class  $\A(\widetilde \bG)$ and $V^G_t = e^{U^G_t}$.
\end{theorem}

\begin{proof}
For any ${\boldsymbol{\theta}} \in \A(\widetilde \bG)$, we apply the It\^o product rule to compute $e^{-\alpha X_t^\theta+U_t^G}$, for every $t \in [0,T]$.
By \eqref{eq:BSDEq-e}, we get that
\begin{equation}\label{eq:decomp}
\begin{split}
& \ud \left(e^{-\alpha X_t^\theta+U_t^G}\right)  \\
& = \ud M_t^{U,\theta} + e^{-\alpha X_t^\theta}\left\{f(t, \gamma_t^1e^{U_{t-}^G},\gamma_t^2e^{U_{t-}^G},\gamma_t^3e^{U_{t-}^G},e^{U_{t-}^G})  - f_\alpha(t,\gamma_t^1e^{U_{t-}^G},\gamma_t^2e^{U_{t-}^G},\gamma_t^3e^{U_{t-}^G},e^{U_{t-}^G},\theta_t^1,\theta_t^2)\right\}\ud t,
\end{split}
\end{equation}
where 
the function $f$ 
is given by
\begin{align}
& f(t,\gamma_t^1e^{U_{t-}^G},\gamma_t^2e^{U_{t-}^G},\gamma_t^3e^{U_{t-}^G},e^{U_{t-}^G})\\
&=e^{U_{t-}^G}\left[\widetilde f(t,\gamma_t^1,\gamma_t^2,\gamma_t^3,\gamma_t^4,U_{t-}^G)+\frac{1}{2}\left((\gamma_t^1)^2+(\gamma_t^2)^2+(\gamma_t^3)^2\right)+\left(e^{\gamma_t^4}-1-\gamma_t^4\right)\pi_t(\lambda)(1-H_t)\right]\\
&= \frac{1}{2} e^{U_{t-}^G}\left(\frac{\mu^S(t, Y_t)}{\sigma^S(t, Y_t)}+ \gamma^1_t\right)^2+ \frac{1}{2} e^{U_{t-}^G}\frac{\left(\mu^B(t, \mu_t, Y_t)+c^B(t,\mu_t,Y_t)\gamma_t^2+d^B(t,\mu_t,Y_t) \gamma_t^3\right)^2}{c^B(t,\mu_t,Y_t)^2+d^B(t,\mu_t,Y_t)^2},
\end{align}
and the function $f_\alpha$ is given by
\begin{align}
&f_\alpha(t,r^1,r^2,r^3,v,\theta^1,\theta^2) \\
& = \alpha \ v \left[\theta_t^1\mu^S(t,Y_t)+\theta_t^2\mu^B(t,\mu_t,Y_t)\right] \\
&\quad + \alpha\left[r^1\theta_t^1\sigma^S(t,Y_t) + r^2 \theta_t^2c^B(t,\mu_t,Y_t) + r^3\theta_t^2d^B(t,\mu_t,Y_t)\right]\\
& \quad \quad - \frac{1}{2} \alpha^2 v\left[(\theta_t^1\sigma^S(t,Y_t))^2+(\theta_t^2)^2\left(
c^B(t,\mu_t,Y_t)^2+d^B(t,\mu_t,Y_t)^2\right)\right],\label{def:driver}
\end{align}
while the process $M^{U,\theta} = \{M_t^{U,\theta},\ t \in [0,T]\}$ is defined as
\begin{align}
M_t^{U,\theta} & :=  \int_0^te^{-\alpha X_u^\theta+U_u^G}\left(\gamma_u^1-\alpha \theta_u^1\sigma^S(u,Y_u)\right)\ud W_u^1 \\
& \quad + \int_0^te^{-\alpha X_u^\theta+U_u^G}\left(\gamma_u^2-\alpha \theta_u^2c^B(u,\mu_u,Y_u)\right)\ud W_u^2\\
& \quad + \int_0^te^{-\alpha X_u^\theta+U_u^G}\left(\gamma_u^3-\alpha \theta_u^2d^B(u,\mu_u,Y_u)\right)\ud W_u^3 + \int_0^te^{-\alpha X_u^\theta+U_u^G}(e^{\gamma_u^4}-1) \ud M_u^\tau,\label{def:MU}
\end{align}
for every $t \in [0,T]$, and it is a $(\widetilde \bG,\P)$-local martingale.
Let us observe that  for any ${\boldsymbol{\theta}} \in \A(\widetilde \bG)$,
\begin{align}
f(t,\gamma_t^1e^{U_{t-}^G},\gamma_t^2e^{U_{t-}^G},\gamma_t^3e^{U_{t-}^G},e^{U_{t-}^G}) \!  & = \! {{\rm ess}\sup}_{{\boldsymbol{\overline \theta}} \in \mathcal A(\widetilde \bG)}
 f_\alpha(t,\gamma_t^1e^{U_{t-}^G},\gamma_t^2e^{U_{t-}^G},\gamma_t^3e^{U_{t-}^G},e^{U_{t-}^G},\overline \theta_t^1, \overline \theta_t^2) \\
 & \geq
 f_\alpha(t,\gamma_t^1e^{U_{t-}^G},\gamma_t^2e^{U_{t-}^G},\gamma_t^3e^{U_{t-}^G},e^{U_{t-}^G},\theta_t^1,\theta_t^2).
\end{align}
Hence,  $$\ud A_t^{\theta}:=  e^{-\alpha X_t^\theta}\left\{f(t, \gamma_t^1e^{U_{t-}^G},\gamma_t^2e^{U_{t-}^G},\gamma_t^3e^{U_{t-}^G},e^{U_{t-}^G})  - f_\alpha(t,\gamma_t^1e^{U_{t-}^G},\gamma_t^2e^{U_{t-}^G},\gamma_t^3e^{U_{t-}^G},e^{U_{t-}^G},\theta_t^1,\theta_t^2)\right\}\ud t$$ is an increasing process and by \eqref{eq:decomp}, we get
\begin{equation}\label{eq:M2}e^{-\alpha X_t^\theta+U_t^G} = e^{-\alpha x_0+U_0^G} + M^{U,\theta}_t + A_t, \quad t \in [0,T].\end{equation}

Since $U^G$ is bounded and $M^{U,\theta}$ given in \eqref{def:MU} is a $(\widetilde \bG,\P)$-local martingale, denoting by $\{\tau_n\}$ a localizing sequence  and  by using \eqref{eq:M2}, we get
$$
\esp{A_{T\wedge\tau_n}^{\theta}} \leq C  \big ( \esp{ e^{-\alpha X^\theta_{T\wedge\tau_n}}} + 1\big ),
$$
for some positive constant $C$. Finally, since the family \eqref{eq:integ_3} is $\P$-uniformly integrable, we have that $\esp{A_{T}^{\theta}} < \infty$ and
that $M^{U,\theta}$ is a $(\widetilde \bG,\P)$-martingale for every ${\boldsymbol{\theta}} \in \A(\widetilde \bG)$.  Hence  $\{e^{-\alpha X_t^\theta+U_t^G},\ t \in [0,T]\}$, for any $\theta \in \bar  \A(\widetilde \bG)$, is a $(\widetilde \bG,\P)$-submartingale and by Proposition \ref{VER}  we get that $V^G_t \geq e^{U^G_t}$, for every $t \in [0,T]$, $\P$-a.s..

By Dol\'eans-Dade formula it follows
$$
\E \left( \int_0^. (e^{\gamma^4_s} -1)  \ud M^\tau_s  \right)_t =e^{\int_0^t \gamma^4_s \ud M^\tau_s - \int_0^t (e^{\gamma^4_s} - 1- \gamma^4_s) (1 -H_s) \pi_s(\lambda) \ud s},
$$
and
\begin{align}
e^{U_t^G  - U_0^G}  = & \E \left ( \int_0^. (e^{\gamma^4_s} -1)  \ud M^\tau_s  + \int_0^. \sum_{i=1}^3\gamma^i_s \ud W^i_s \right)_t  e^{\frac{1}{2} \int_0^t \left(\frac{\mu^S(s, Y_s)}{\sigma^S(s, Y_s)}+ \gamma_s^1\right)^2 \ud s}\\
& e^{\frac{1}{2} \int_0^t \left[ \frac{(\mu^B(s,\mu_s,Y_s) )^2 + (c^B(s, \mu_s,Y_s) \gamma^2_s  + d^B(s, \mu_s,Y_s) \gamma^3_s)^2 }{c^B(s, \mu_s,Y_s)^2 + d^B(s, \mu_s,Y_s)^2}\right] \ud s}.\label{urra}
\end{align}
Then, the stochastic exponential $\left\{\E \left( \int_0^. (e^{\gamma^4_s} -1)  \ud M^\tau_s  + \int_0^. \sum_{i=1}^3\gamma^i_s \ud W^i_s \right)_t,\ t \in [0,T]\right\}$ is a $(\widetilde \bG, \P)$-bounded martingale since $U^G$ is bounded.

Finally, if ${\boldsymbol{\theta}}^* \in \A(\widetilde \bG)$ one obtains that $\{e^{-\alpha X_t^{\boldsymbol{\theta}^*} +U_t^G},\ t \in [0,T]\}$ is a $(\widetilde \bG,\P)$-martingale. Hence, again by Proposition \ref{VER}, we get that $V^G_t =e^{U^G_t}$ and ${\boldsymbol{\theta}}^*$ is an optimal control.



\end{proof}

\begin{proposition}\label{prop:suff_cond}
The strategy ${\boldsymbol{\theta}}^*=(\theta_t^{1,*},\theta_t^{2,*})^\top$ defined by \eqref{theta1*}-\eqref{theta2*} satisfies integrability conditions \eqref{eq:integ_1} and \eqref{eq:integ_2} of
Definition \ref{def:selffinancing}.


Moreover, if the functions $\ds \frac{\mu^S(t,y)} {\sigma^S(t,y)}$  and $\ds \frac{\mu^B(t,\mu,y)} {\sqrt{ (c^B(t,\mu,y))^2 +  (d^B(t,\mu,y))^2}}$ are bounded, the family \eqref{eq:integ_3} in Definition \ref{def:selffinancing} is $\P$-uniformly integrable and
${\boldsymbol{\theta}}^*\in \mathcal{A}(\widetilde{\bG})$.
\end{proposition}
\begin{proof}
In the first part of the proof we show that
 $$
 \int_0^T \left\{( \theta^{1,*}_t \sigma^S(t, Y_t))^2 +  (\theta^{2,*}_t)^2 (c^B(t, \mu_t, Y_t)^2 + d^B(t, \mu_t, Y_t)^2 )\right\} \ud t< \infty, \quad \P-\mbox{a.s.}.
 $$
By Assumption \ref{mgloc}, we get that $\int_0^T (\theta_s^{1,*}\sigma^S(s,Y_s))^2\ud s < \infty$, $\P$-a.s., since $\gamma^1 \in L^2(W;\widetilde \bG)$ and
\begin{align}
& \int_0^T  (\theta^{2,*}_s)^2 \left(c^B(s, \mu_s, Y_s)^2 + d^B(s, \mu_s, Y_s)^2 \right) \ud s \\
&\leq  \frac{2}{\alpha^2} \int_0^T \left(\mu^B(s, \mu_s, Y_s)^2+ \frac {\left(c^B(s, \mu_s,Y_s) \gamma^2_s + d^B(s, \mu_s,Y_s)  \gamma^3_s\right)^2} {c^B(s, \mu_s,Y_s)^2 +d^B(s, \mu_s,Y_s)^2}\right)\ud s \\
& \leq \frac{4}{\alpha^2} \left\{\int_0^T \mu^B(s, \mu_s, Y_s)^2 \ud s+ \int_0^T \frac {c^B(s, \mu_s,Y_s)^2} {c^B(s, \mu_s,Y_s)^2 + d^B(s, \mu_s,Y_s)^2} (\gamma^2_s)^2\ud s \right.\\
& \quad\left. + \int_0^T\!\!\!\! \frac {d^B(s, \mu_s,Y_s)^2} {c^B(s, \mu_s,Y_s)^2 + d^B(s, \mu_s,Y_s)^2} (\gamma^3_s)^2\ud s\right\}\\
&\leq \frac{4}{\alpha^2} \left\{\int_0^T \mu^B(s, \mu_s, Y_s)^2 \ud s+ \int_0^T (\gamma^2_s)^2\ud s + \int_0^T (\gamma^3_s)^2\ud s\right\}< \infty,\quad \P-\mbox{a.s},
\end{align}
where the last inequality holds since $\gamma^2,\ \gamma^3 \in L^2(W;\widetilde \bG)$. Moreover, thanks to Assumption \ref{mgloc} and since $\gamma^1,\ \gamma^2,\ \gamma^3 \in L^2(W;\widetilde \bG)$, it is easy to check that the integrability condition $\int_0^T\{|\theta^{1,*}_t\mu^S(t, Y_t)|+|\theta^{2,*}_t\mu^B(t,\mu_t,Y_t)|\}\ud t< \infty$, $\P$-a.s., 
is satisfied.

For the second part of the statement,  using a direct computation we get that 
$$
e^{- \alpha X^{\theta^*}_t}  =  e^{-\alpha x_0} e^ {-(U_t^G  - U_0^G)}\E \left( \int_0^. (e^{\gamma^4_s} -1)  \ud M^\tau_s  - \int_0^. \sum_{i=1}^3\Gamma^i_s \ud W^i_s \right)_t,
$$
 for all $t \in [0,T]$, where we set
 \begin{align*}
 \Gamma^1_t &= \frac{\mu^S(t, Y_t)} {\sigma^S(t, Y_t)},\\
   \Gamma^2_t & = \frac{\mu^B(t, \mu_t, Y_t) c^B(t, \mu_t, Y_t) } {(c^B(t, \mu_t, Y_t))^2 + (d^B(t, \mu_t, Y_t))^2} + \frac{d^B(t, \mu_t, Y_t)(c^B(t, \mu_t, Y_t) \gamma^3_t - d^B(t, \mu_t, Y_t) \gamma^2_t)}{(c^B(t, \mu_t, Y_t)^2 + (d^B(t, \mu_t, Y_t))^2},   \\
 \Gamma^3_t & = \frac{\mu^B (t, \mu_t, Y_t)d^B(t, \mu_t, Y_t) } {(c^B(t, \mu_t, Y_t))^2 + (d^B(t, \mu_t, Y_t))^2} + \frac{c^B(t, \mu_t, Y_t)(
 d^B(t, \mu_t, Y_t) \gamma^2_t - c^B(t, \mu_t, Y_t) \gamma^3_t )}{(c^B(t, \mu_t, Y_t))^2 + (d^B(t, \mu_t, Y_t))^2}.
 \end{align*}
Assume that $\ds \frac{\mu^S(t,y)} {\sigma^S(t,y)}$  and $\ds \frac{\mu^B(t,\mu,y)} {\sqrt{ (c^B(t,\mu,y))^2 +  (d^B(t,\mu,y))^2}}$ are bounded. By the Dol\'eans-Dade formula we obtain that
\begin{align}
&\E \left( \int_0^. (e^{\gamma^4_s} -1)  \ud M^\tau_s  - \int_0^. \sum_{i=1}^3\Gamma^i_s \ud W^i_s \right)_t \nonumber\\
&\quad = \E \left( \int_0^. (e^{\gamma^4_s} -1)  \ud M^\tau_s \right)_t \E \left(-  \int_0^. \Gamma^1_s \ud W^1_s\right)_t  \E \left( - \int_0^.  \Gamma^2_s \ud W^2_s  -  \int_0^.  \Gamma^3_s \ud W^3_s \right)_t.\label{eq:expon}
\end{align}
The first two stochastic exponentials on the right-hand side of equation \eqref{eq:expon} are  uniformly integrable $(\widetilde \bG,\P)$-martingales. We now prove that the third stochastic exponential is also a uniformly integrable $(\widetilde \bG,\P)$-martingale.


Let $\widetilde W^i$, $i=1,2$ be the $(\widetilde \bG,\P)$-Brownian motions given  by
 $$\ud \widetilde W^1_s = \frac{c^B (s, \mu_s, Y_s) } {\sqrt{c^B(s, \mu_s, Y_s)^2 + d^B(s, \mu_s, Y_s)^2 }} \ud W^2_s +  \frac{d^B (s, \mu_s, Y_s) } {\sqrt{c^B(s, \mu_s, Y_s)^2 + d^B(s, \mu_s, Y_s)^2 }} \ud W^3_s,$$
 $$\ud \widetilde W^2_s = \frac{ - d^B (s, \mu_s, Y_s) } {\sqrt{c^B(s, \mu_s, Y_s)^2 + d^B(s, \mu_s, Y_s)^2 }} \ud W^2_s + \frac{c^B (s, \mu_s, Y_s) } {\sqrt{c^B(s, \mu_s, Y_s)^2 + d^B(s, \mu_s, Y_s)^2 }} \ud W^3_s.
 $$
Then, we can rewrite the stochastic exponential as
\begin{align}
\E \left( - \int_0^.  \Gamma^2_s \ud W^2_s  -  \int_0^.  \Gamma^3_s \ud W^3_s \right)_t & = \E \left( - \int_0^.  \frac{\mu^B (s, \mu_s, Y_s) } {\sqrt{c^B(s, \mu_s, Y_s)^2 + d^B(s, \mu_s, Y_s)^2 }} \ud \widetilde W^1_s \right)_t \nonumber \\
&\times e^{\int_0^t  \frac{ d^B (s, \mu_s, Y_s) c^B(s, \mu_s, Y_s)} { c^B(s, \mu_s, Y_s)^2 + d^B(s, \mu_s, Y_s)^2 } \gamma^2_s \gamma^3_s \ud s} \nonumber \\
& \times \E \left( - \int_0^. \frac{ \gamma^2_s d^B (s, \mu_s, Y_s) } {\sqrt{c^B(s, \mu_s, Y_s)^2 + d^B(s, \mu_s, Y_s)^2 }} \ud \widetilde W^2_s \right)_t \nonumber \\
&\times \E \left( \int_0^. \frac{ \gamma^3_s c^B (s, \mu_s, Y_s) } {\sqrt{c^B(s, \mu_s, Y_s)^2 + d^B(s, \mu_s, Y_s)^2 } }\ud \widetilde W^2_s \right)_t. \label{urra2}\end{align}
The stochastic exponentials on the right-hand side of  equality \eqref{urra2} are $(\widetilde \bG,\P)$-uniformly integrable martingales. Moreover, by equation \eqref{urra} and  since
\begin{align*}
&  \ds e^{\int_0^t  \frac{ d^B (s, \mu_s, Y_s) c^B(s, \mu_s, Y_s)} { c^B(s, \mu_s, Y_s)^2 + d^B(s, \mu_s, Y_s)^2 } \gamma^2_s \gamma^3_s \ud s}  \leq
e^{\frac{1}{2} \int_0^t\frac{(c^B(s, \mu_s,Y_s) \gamma^2_s  + d^B(s, \mu_s,Y_s) \gamma^3_s)^2 } {c^B(s, \mu_s,Y_s)^2 + d^B(s, \mu_s,Y_s)^2} \ud s} \\
&\qquad\leq
e^{U_t^G  - U_0^G}  \E \left( \int_0^. (e^{\gamma^4_s} -1)  \ud M^\tau_s  + \int_0^. \sum_{i=1}^3\gamma^i_s \ud W^i_s \right)_t^{-1},\end{align*}
we have that
 $$\esp{  \sup_{s \in[0,T] } e ^{\int_0^t  \frac{ d^B (s, \mu_s, Y_s) c^B(s, \mu_s, Y_s)} { c^B(s, \mu_s, Y_s)^2 + d^B(s, \mu_s, Y_s)^2 } \gamma^2_s \gamma^3_s \ud s} } < \infty,$$
which concludes the proof.
\end{proof}

In view of Theorem \ref{verification} it is crucial to provide existence and uniqueness of the solution to BSDE \eqref{eq:BSDEq-e}.  Therefore, we conclude this section by discussing Examples \ref{ex11} and \ref{C11},  for which  the BSDE has a simpler form and existence and uniqueness of the solution can be derived by applying known results in the literature, see, e.g. \cite[Theorem 3.5]{becherer2006bounded} and \cite[Theorem 11.1.1]{delong2013backward}.


First, we study the setting of Example \ref{ex11} where the risky asset price process and the longevity bond price process are not affected by the stochastic factor $Y$. In this case equation \eqref{eq:BSDEq-e}  reduces to
\begin{align}
U^G_t = \alpha G_T - \int_t^T \sum_{i=1}^2\gamma^i_s \ud W^i_s - \int_t^T \gamma^4_s  \ud M^\tau_s - \int_t^T h(s, \gamma_s^1, \gamma_s^2, \gamma^4_s)  \ud s, \label{eq:BSDEq-e1}
\end{align}
where
\begin{align}
& h(t, \gamma^1, \gamma^2, \gamma^4) = - (e^{\gamma^4} - \gamma^4 - 1) (1 - H_t) \pi_{t}(\lambda)\\
& \qquad +\frac{1}{2}\left\{\left(\frac{\mu^S(t)}{\sigma^S(t)}\right)^2+\left(\frac{\mu^B(t, \mu_t)}{c^B(t, \mu_t)}\right)^2+2\frac{\mu^S(t)}{\sigma^S(t)}\gamma^1 + 2\frac{\mu^B(t, \mu_t)}{c^B(t, \mu_t)}\gamma^2\right\}.
\end{align}
Here, under the assumption that the functions $\ds \frac{\mu^S(t,y)}{\sigma^S(t,y)}$  and $\ds \frac{\mu^B(t,\mu)}{c^B(t,\mu)}$ are bounded,  existence and uniqueness of the solution $(U^G, \gamma^1,\gamma^2, \gamma^4)$ with  $U^G$ bounded, $\gamma^i \in L^2(W;\widetilde \bG), i=1,2$,
$\gamma^4 \in L^1(M^{\tau})$ is proved in \citet[Theorem 4.1]{becherer2006bounded}. Specifically, by Theorem \ref{verification} and Proposition \ref{prop:suff_cond} we get that
\begin{align}
{\boldsymbol{\theta}}^*=(\theta_t^{1,*},\theta_t^{2,*})^\top =\left(  \frac {\mu^S(t) } { \alpha \sigma^S(t)^2}+\frac {\gamma_t^1 } { \alpha \sigma^S(t)},\  \frac { \mu^B(t, \mu_t) } { \alpha c^B(t, \mu_t)^2 }+\frac {\gamma^2_t } { \alpha  c^B(t, \mu_t)}\right)^\top,
\end{align}
for every $t \in [0,T]$, belongs to $\A(\widetilde \bG)$ and then it is an optimal investment strategy. Notice that requiring that functions $\ds \frac{\mu^S(t,y)}{\sigma^S(t,y)}$  and $\ds \frac{\mu^B(t, \mu)}{c^B(t, \mu)}$ are bounded corresponds to have bounded  market prices of risk.

Second, we investigate the optimal value and the optimal strategies in the framework of Example \ref{C11}. Here, the log-value process is the solution of the BSDE
\begin{align}
&U^G_t = \alpha G_T - \int_t^T \gamma^1_s \ud W^1_s - \int_t^T \gamma^3_s \ud W^3_s- \int_t^T   \gamma^4_s  \ud M^\tau_s - \int_t^T  \widetilde h(s, \gamma^1_s,   \gamma^3_s, \gamma^4_s ) \ud s ,
\end{align}
where
\begin{align*}
&\widetilde h(t,\gamma^1, \gamma^3,  \gamma^4 )=  -(e^{\gamma^4} - \gamma^4 - 1) (1 - H_t) \pi_{t}(\lambda)\\
& \qquad +\frac{1}{2}\left\{\left(\frac{\mu^S(t, Y_t)}{\sigma^S(t,Y_t)}\right)^2+\left(\frac{\mu^B(t, Y_t)}{d^B(t, Y_t)}\right)^2+2\frac{\mu^S(t, Y_t)}{\sigma^S(t, Y_t)}\gamma^1 + 2\frac{\mu^B(t, Y_t)}{c^B(t, Y_t)}\gamma^3\right\}
\end{align*}
and the optimal strategy ${\boldsymbol{\theta}}^*=(\theta_t^{1,*},\theta_t^{2,*})^\top \in \A(\widetilde \bG)$, is
\begin{align}
\theta_t^{1,*} &=  \frac{\mu^S(t, Y_t) }{ \alpha \sigma^S(t,Y_t)^2}+\frac { \gamma_t^1 } { \alpha \sigma^S(t,Y_t)}, \quad t \in [0,T],\\
\theta_t^{2,*}& = \frac{\mu^B(t, Y_t) }{ \alpha d^B(t,Y_t)^2}+\frac { \gamma_t^3 } { \alpha d^B(t,Y_t)}, \quad t \in [0,T].
\end{align}
Under the assumption that the market prices of risk $\ds \frac{\mu^S}{\sigma^S}$  and $\ds \frac{\mu^B}{d^B}$ are bounded, we still have existence and uniqueness of the solution $(U^G, \gamma^1,\gamma^3, \gamma^4)$ with  $U^G$ bounded, $\gamma^i \in L^2(W;\widetilde \bG), i=1,3$,
$\gamma^4 \in L^1(M^{\tau})$.


\subsection{The pure investment problem} \label{sec:bsde2}

First, note that Problem \ref{p2} corresponds to a special case of Problem \ref{p1}, choosing $G_T=0$ and available information level given by $\widetilde \bF$.
Therefore, we solve Problem \ref{p2} by applying similar techniques to those given in the previous subsection and characterize the log value process $\log V^0$ and the optimal investment strategy in terms of the solution of another BSDE, which, in this case, has quadratic driver.

\begin{theorem}\label{verification0}
Let $(U^0, \phi^1, \phi^2, \phi^3)$, with $U^0$ bounded, $\phi^i \in L^2(W;\widetilde \bF)$, $i=1,2,3$,
 be a solution to the BSDE
\begin{equation}\label{eq:BSDEver0}
 U^0_t =  - \int_t^T \sum_{i=1}^3\phi^i_s \ud W^i_s  - \int_t^T \widetilde f^0(s,\phi^1_s,  \phi^2_s, \phi^3_s ) \ud s,
\end{equation}
where
\begin{align}
 \widetilde f^0(t,\psi^1, \psi^2, \psi^3)=& -  \frac{1}{2} \left((\psi^1)^2+(\psi^2)^2+( \psi^3)^2\right) + \frac{1}{2} \left(\frac{\mu^S(t, Y_t)}{\sigma^S(t, Y_t)}+\psi^1\right)^2\\
& +\frac{1}{2}  \frac{\left(\mu^B(t, \mu_t, Y_t)+c^B(t,\mu_t,Y_t)\psi^2+d^B(t,\mu_t,Y_t) \psi^3\right)^2}{(c^B(t,\mu_t,Y_t))^2+(d^B(t,\mu_t,Y_t))^2}. \label{eq:generator0}
\end{align}
Then, $V^0_t \geq e^{U^0_t}$, for every $t \in [0,T]$, $\P$-a.s. and
$$
\left\{\E \left( \int_0^. \sum_{i=1}^3\phi^i_s \ud W^i_s \right)_t,\  t \in [0,T]\right\}
$$
is a bounded $(\widetilde \bG,\P)$-martingale.
Moreover, if ${\boldsymbol{\vartheta}}^*=(\vartheta_t^{1,*},\vartheta_t^{2,*})^\top \in \A(\widetilde \bF)$, with
\begin{align}
\vartheta_t^{1,*} &=  \frac {\mu^S(t, Y_t) } { \alpha \sigma^S(t,Y_t)^2}+\frac {\phi_t^1 } { \alpha \sigma^S(t,Y_t)},\\
\vartheta_t^{2,*}& = \frac { \mu^B(t, \mu_t,Y_t) + c^B(t, \mu_t,Y_t) \phi^2_t + d^B(t, \mu_t,Y_t)  \phi^3_t} { \alpha[  (c^B(t, \mu_t,Y_t) )^2 + (d^B(t, \mu_t,Y_t) )^2]},
\end{align}
for every $t \in [0,T]$, then  $V^0_t = e^{U^0_t}$ and ${\boldsymbol{\vartheta}}^*$ is an optimal strategy.
\end {theorem}

The proof follows the same lines as that of Theorem \ref{verification}.

Analogously to Proposition \ref{prop:suff_cond},  under the hypothesis that the functions $\displaystyle \frac{\mu^S(t,y)}{\sigma^S(t,y)}$  and $\displaystyle \frac{\mu^B(t,\mu,y)}{\sqrt{ (c^B(t,\mu,y))^2 +  (d^B(t,\mu,y))^2}}$ are bounded, we get that ${\boldsymbol{\vartheta}}^*=(\vartheta_t^{1,*},\vartheta_t^{2,*})^\top \in \A(\widetilde \bF)$. We also observe that equation \eqref{eq:BSDEver0} is a BSDE with quadratic generator given by \eqref{eq:generator0}.
In this case, existence and uniqueness of the solution are provided, for instance, in \citet[Theorem 7.3.3]{zhang2017backward}, whereas existence and uniqueness of the solution of the BSDE \eqref{eq:BSDEq-e} is studied in the next section.

\section{The indifference price of the pure endowment} \label{sec:indiff-price}

We recall that the value processes $V^G$ and $V^0$ corresponding to the investment problem with and without the derivative, coincide for all $T>t>\tau$. This, in turn impliest that, in order to compute the indifference price given in \eqref{eq:indiff-price}, we only need to study the solution of  BSDE \eqref{eq:BSDEq-e} over the stochastic interval  $\llbracket 0, \tau \wedge T\rrbracket$. Since $G_T=\xi \I_{\{\tau > T\}}$, see  Definition \ref{eq:pure endowment}, this corresponds to consider the following BSDE with random time horizon
\begin{align}
& U^G_t \! = \alpha \xi \I_{\{ \tau > T\}} \! - \! \int_{t\wedge \tau}^{T\wedge \tau} \!\! \sum_{i=1}^3 \gamma^i_s \ud W^i_s \! -\! \int_{t\wedge \tau}^{T \wedge \tau}\!\!\!\! \gamma^4_s  \ud M^\tau_s  \!- \! \int_{t\wedge \tau}^{T\wedge \tau} \!\!\!\!\widetilde f(s,\gamma^1_s,  \gamma^2_s,  \gamma^3_s,  \gamma^4_s ) \ud s,\qquad{} \label{eq:BSDEq-eRH}
\end{align}
for every $t \in [0,T]$,
which is equivalent to equation \eqref{eq:BSDEq-e} over the stochastic time interval $\llbracket 0, \tau \wedge T\rrbracket$.
According to  \citet[Theorem 4.3]{kharroubi2013mean} and \citet[Proposition 4.1]{jeanblanc2015utility}, we introduce a BSDE in the Brownian filtration $\widetilde \bF$, stopped at $\tau$, and establish an equivalence result with the solution of BSDE \eqref{eq:BSDEq-eRH} given in Lemma \ref{ex} below.\\


We recall that on the time interval  $\{ t < \tau\wedge T\}$ the process $\pi(\lambda)$ coincides  with the $\widetilde \bF$-adapted process  $\widehat  \pi(\lambda)$ given by \eqref{eq:hat_pi}.

\begin{lemma}\label{ex}
Let  $(\widehat U, \widehat \gamma^1, \widehat \gamma^2, \widehat \gamma^3)$, where $\widehat U$ is $\widetilde \bF$-adapted and bounded, and $\widehat \gamma^i \in L^2(W;\widetilde \bF)$, $i=1,2,3$,
 be a solution to the BSDE
 \begin{align}
&\widehat U_t = \alpha \xi - \int_t^T \sum_{i=1}^3\widehat \gamma^i_s \ud W^i_s -\int_t^T \left(\frac{1}{2} \left(\frac{\mu^S(s, Y_s)}{\sigma^S(s, Y_s)}\right)^2+\frac{\mu^S(s, Y_s)}{\sigma^S(s, Y_s)} \widehat \gamma^1_s\right) \ud s\\
& - \int_t^T  \frac{\frac{1}{2}(\mu^B(s,\mu_s, Y_s))^2+ \mu^B(s,\mu_s, Y_s)\left(c^B(s,\mu_s,Y_s)\widehat \gamma^2_s+d^B(s,\mu_s,Y_s)\widehat \gamma^3_s\right)}{c^B(s,\mu_s,Y_s)^2+d^B(s,\mu_s,Y_s)^2} \ud s,\\
&  + \int_t^T \left(\left(e^{ -\widehat U_s} - 1\right)\widehat \pi_{s}(\lambda)   +   \frac{(d^B(s,\mu_s,Y_s)\widehat \gamma^2_s -  c^B(s,\mu_s,Y_s)\widehat \gamma^3_s)^2}{ 2 [c^B(s,\mu_s,Y_s)^2+d^B(s,\mu_s,Y_s)^2] }  \right) \ud s\label{eq:BSDE_BM}
\end{align}
for every $ t \in[0,T]$,
where $\widehat \pi(\lambda)$ is given by \eqref{eq:hat_pi}. Then, $(U^G, \gamma^1, \gamma^2,  \gamma^3, \gamma^4)$ defined as
$$
U^G_t = \widehat U_t \I_{\{ t < \tau\}}, \quad  \gamma^i_t = \widehat \gamma^i_t \I_{\{ t < \tau\}}, \quad i = 1,2,3 \quad
\gamma^4_t  = - \widehat U_{t^-} \I_{\{ t \leq \tau\}},
 $$
is a solution of the BSDE \eqref{eq:BSDEq-eRH}, where $\gamma^i \in L^2(W;\widetilde \bG)$, $i=1,2,3$, $U^G$ is $\widetilde
\bG$-adapted and bounded, and $\gamma^4  \in L^1(M^{\tau})$.
\end{lemma}

\begin{proof}
To get the result, we apply the It\^o product rule to $U^G_t = \widehat U_t \I_{\{ t < \tau\}}=  \widehat U_t (1-H_t)$  and observe that  $U^G_{T\wedge \tau} = \alpha \xi  \I_{\{  \tau > T \}}$.
\end{proof}

\noindent By Lemma \ref{ex}, it is clear that existence and uniqueness of the solution of BSDE \eqref{eq:BSDEq-eRH} follows from existence and uniqueness of the solution of equation \eqref{eq:BSDE_BM}, which is a quadratic-exponential BSDE, only driven by Brownian motions. Using this argument we get the following result.

\begin{proposition}\label{exist}
Assume that the functions $\displaystyle \frac{\mu^S(t,y)}{\sigma^S(t,y)}$  and $\displaystyle \frac{\mu^B(t,\mu,y)}{\sqrt{ (c^B(t,\mu,y))^2 +  (d^B(t,\mu,y))^2}}$ are bounded. Then, there exists a unique solution $(U^G, \gamma^1, \gamma^2, \gamma^3, \gamma^4)$ to BSDE \eqref{eq:BSDEq-eRH} where  $\gamma^i \in L^2(W;\widetilde \bG)$, $i=1,2,3$, $U^G$ is $\widetilde
\bG$-adapted and  bounded, $\gamma^4  \in L^1(M^{\tau})$ such that
$$
\int_0^t \sum_{i=1}^3\gamma^i_s \ud W^i_s + \int_0^t ( e ^{\gamma^4_s} -1) \ud M^\tau_s
$$
is a BMO$(\widetilde \bG)$-martingale.
\end{proposition}

\begin{proof}
Existence and uniqueness of the solution $(\widehat U, \widehat \gamma^1, \widehat \gamma^2, \widehat \gamma^3)$, with $\widehat U \in L^2(W;\widetilde \bF)$ bounded, and $\widehat \gamma^i \in L^2(W;\widetilde \bF), i=1,2,3$, to equation \eqref{eq:BSDE_BM} follow from the same argument used in the proof of \citet[Theorem 4.1]{jeanblanc2015utility}.
Precisely, by Lemma \ref{lemma:density} in Appendix \ref{appendix:tech_res} and boundedness of $\lambda$,  hypotheses (H1) and (H2) in \citet[Section 2.2]{jeanblanc2015utility} hold. Moreover, the driver is of the form
$$
(e^{-u}-1)\widehat \pi(\lambda) - \widehat g(t,\widehat \gamma^1,\widehat \gamma^2, \widehat \gamma^3),
$$
where $g$ is a map from $[0,T] \times \R \times \R$ to $\R$ defined as
\begin{align*}
&g(t,\widehat \gamma^1, \widehat \gamma^2, \widehat \gamma^3)\!=\frac{(d^B(t,\mu_t,Y_t)\widehat \gamma^2 -  c^B(t,\mu_t,Y_t)\widehat \gamma^3)^2}{ 2 [c^B(t,\mu_t,Y_t)^2+d^B(t,\mu_t,Y_t)^2] } \!+\! \frac{1}{2} \!\left(\frac{\mu^S(t, Y_t)}{\sigma^S(t, Y_t)}\right)^2\!\!+\frac{\mu^S(t, Y_t)}{\sigma^S(t, Y_t)} \widehat \gamma^1\\
& \qquad + \frac{\frac{1}{2}(\mu^B(t,\mu_t, Y_t))^2 + \mu^B(t,\mu_t, Y_t)\left(c^B(t,\mu_t,Y_t)\widehat \gamma^2+d^B(t,\mu_t,Y_t)\widehat \gamma^3_t\right) }{c^B(t,\mu_t,Y_t)^2+d^B(t,\mu_t,Y_t)^2}.
\end{align*}
For every $(\widehat \gamma^1,\widehat \gamma^2, \widehat \gamma^3) \in \R \times \R\times \R$, $g(\cdot,\widehat \gamma^1,\widehat \gamma^2, \widehat \gamma^3)$ is $\widetilde \bF$-progressively measurable.
It is also easy to check that for $(\widehat \gamma^1,\widehat \gamma^2, \widehat \gamma^3)=(0,0,0)$, $g(t,0, 0,0)=0$, for every $t \in [0,T]$, and that $g$ is Lipschitz with respect to $\widehat \gamma^1$, $\widehat \gamma^2$ and $\widehat \gamma^3$, which imply that Assumption 4.1 in \citet{jeanblanc2015utility} holds.
Then, by Lemma \ref{ex} we get existence of  a solution to BSDE \eqref{eq:BSDEq-eRH} and \citet[Lemma 4.1]{jeanblanc2015utility} yields uniqueness.
\end{proof}

Finally, by gathering the results we have the following characterization of the indifference price process of the pure endowment introduced in \eqref{eq:pure endowment}.
\begin{proposition}\label{prop:prezzo}
Assume that the functions $\displaystyle \frac{\mu^S(t,y)} {\sigma^S(t,y)}$  and $\displaystyle \frac{\mu^B(t,\mu,y)}{\sqrt{ (c^B(t,\mu,y))^2 +  (d^B(t,\mu,y))^2}}$ are bounded. Let $U^0$ be the unique bounded and $\widetilde \bF$-adapted solution to equation \eqref{eq:BSDEver0} and let $\widehat U$ be the unique bounded $\widetilde \bF$-adapted solution to equation \eqref{eq:BSDE_BM}. Then, the indifference price $p^\alpha$ of the pure endowment is given by
\begin{equation}
p_t^\alpha= \frac{1}{\alpha} \left(\widehat U_t - U_t^0\right)\I_{\{\tau > t\}}  \quad t \in [0,T].
\end{equation}
\end{proposition}

\subsection{The indifference price and the utility indifference strategy in the market model of Example \ref{ex11}}

In this paragraph, we consider the setting of Example \ref{ex11}. We recall that the dynamics of the risky asset price process $S^1$ and the longevity bond price process $S^2$ are respectively described  by
\begin{align}
 \ud S_t^1 & = S_t^1\left(\mu^S(t)\ud t +\sigma^S(t) \ud W^1_t\right), \quad S_0^1 = s_0^1 \in \R^+,\\
\ud S^2_t & = S^2_t\left(\mu^B(t, \mu_t)\ud t +c^B(t,\mu_t)\ud W_t^2\right), \quad S_0^2=s_0^2\in \R^+
\end{align}
and the process $\mu$ follows equation \eqref{def:index_intensity}.
Here, the filtration $\widetilde \bF$ corresponds to the natural filtration of Brownian motions $W^1$ and $W^2$. 
We aim to provide a more explicit representation for the indifference price of a pure endowment policy, whose payoff is given by  \eqref{eq:pure endowment}, in this special setting.

We also assume that the functions $\ds \frac{\mu^S(t,y)}{\sigma^S(t,y)}$ and $\ds \frac{\mu^B(t,\mu)}{c^B(t,\mu)}$ are bounded.

We define a probability measure $\overline \P$ equivalent to $\P$ on $\widetilde \G_T$ such that the processes $\overline W^1=\{\overline W_t^1,\ t \in [0,T]\}$, $\overline W^2=\{\overline W_t^2,\ t \in [0,T]\}$ given by
\begin{align}
\overline W_t^1 & := W_t^1 + \int_0^t \frac{\mu^S(u)}{\sigma^S(u)} \ud u, \\
\overline W_t^2 & := W_t^2 + \int_0^t \frac{\mu^B(u,\mu_u)}{c^B(u,\mu_u)} \ud u,
\end{align}
for every $t \in [0,T]$, are $(\widetilde \bG,\overline \P)$-Brownian motions and the $(\widetilde \bG, \overline \P)$-mortality intensity is still given by  $\pi(\lambda)(1-H)$. Note that in particular $\overline W^1$ and $\overline W^2$ are $(\widetilde \bF,\overline \P)$-Brownian motions and the restriction of  $\overline \P$ on $\widetilde \F_T$ represents the unique martingale measure on the complete primary financial-insurance market, given by the money market account, the stock and the longevity bond.

Now, recall that the log-value process $U^G$ solves Equation \eqref{eq:BSDEq-e1}. 
Then, under $\overline \P$ we get that $U^G$ is the solution of the following BSDE
\begin{align}
U^G_t = \alpha G_T - \int_t^T \sum_{i=1}^2\gamma^i_s \ud \overline W^i_s - \int_t^T \gamma^4_s  \ud M^\tau_s - \int_t^T \overline h(s, \gamma^4_s)  \ud s,
\end{align}
for every $t \in [0,T]$, where the function $\overline h(t, \gamma^4)$ is given by
\begin{align}
\overline h(t, \gamma^4) = - (e^{\gamma^4} - \gamma^4 - 1) (1 - H_t) \pi_{t}(\lambda)
+\frac{1}{2}\left(\frac{\mu^S(t)}{\sigma^S(t)}\right)^2+\frac{1}{2}\left(\frac{\mu^B(t, \mu_t)}{c^B(t, \mu_t)}\right)^2.
\end{align}
As in Lemma \ref{ex}, we can restrict to the Brownian filtration $\widetilde \bF$ by introducing the process $\widehat U$, which satisfies the following BSDE
\begin{align}
&\widehat U_t = \alpha \xi -\! \int_t^T \!\!\!\sum_{i=1}^2\widehat \gamma^i_s \ud \overline W^i_s -\frac{1}{2} \int_t^T \!\! \left[\left(\frac{\mu^S(s)}{\sigma^S(s)}\right)^2 \! +\left(\frac{\mu^B(s,\mu_s)}{c^B(s,\mu_s)}\right)^2 \right]\ud s
+\! \int_t^T \!\!\!\left(e^{ -\widehat U_s} - 1\right)\widehat \pi_{s}(\lambda)\ud s,\quad{}\label{eq:Uhat}
\end{align}
for every $ t \in[0,T]$. Furthermore, the log-value process in the pure investment case $U^0$ satisfies
\begin{equation}
 U^0_t =  - \int_t^T \sum_{i=1}^2\phi^i_s \ud \overline W^i_s  - \int_t^T \frac{1}{2} \left[\left(\frac{\mu^S(s)}{\sigma^S(s)}\right)^2 +\left(\frac{\mu^B(s,\mu_s)}{c^B(s,\mu_s)}\right)^2 \right]\ud s,\label{eq:U0}
\end{equation}
for every $ t \in[0,T]$ and then it is given by
\begin{equation}
 U^0_t =  - \frac{1}{2} \mathbb{E}^{\overline \P} \left[ \int_t^T\left[ \left(\frac{\mu^S(s)}{\sigma^S(s)}\right)^2 +\left(\frac{\mu^B(s,\mu_s)}{c^B(s,\mu_s)}\right)^2 \right]\ud s\bigg{|}\widetilde \F_t\right],
\end{equation}
where $\mathbb{E}^{\overline \P}$ denotes the expectation with respect to $\overline \P$.

 \begin{proposition}\label{Bsde_ind}
 The indifference price is given by $p_t^\alpha = \overline p_t^\alpha \I_{\{t < \tau \}}$,  where $(\overline p^\alpha, \overline \gamma^1, \overline \gamma^2) $  is the unique solution to the following BSDE
 \begin{equation}
 \overline p_t^\alpha  = \xi - \int_t^T \sum_{i=1}^2\overline \gamma_s^i \ud \overline W_s^i + \frac{1}{\alpha}\int_t^T \left(e^{-U_s^0}e^{-\alpha \overline p_s^\alpha} - 1\right)\widehat \pi_s(\lambda) \ud s, \label{eq:pbar1}
 \end{equation}
with $\overline p^\alpha$ bounded and $ \overline \gamma^i$ are $\widetilde \bF$-predictable and such that $\mathbb{E}^{\overline \P}\left[\int_0^T|\gamma^i_s|^2 \ud s\right]<\infty$ for  $i=1,2$.

Moreover, for every $t \leq T\wedge \tau$,  the utility indifference hedging strategy is given by
 \begin{align}
\overline\theta_t=\left(\frac{\overline \gamma^1_t}{\sigma^S(t)};\frac{\overline \gamma^2_t}{c^B(t, \mu_t)}\right). \label{eq:hedg-strategy}
\end{align}

 \end{proposition}

\begin{proof}
In view of Proposition \ref{prop:prezzo}, using equations \eqref{eq:Uhat} and \eqref{eq:U0} we get that $\overline p^\alpha_t$ solves the BSDE \eqref{eq:pbar1}
for every $ t \in[0,T]$, where $\overline \gamma^i:= \ds \frac{\widehat \gamma^i-\phi^i}{\alpha}$, with $i=1,2$. Note that standard results on uniqueness for the solution of the BSDE \eqref{eq:pbar1} follow by uniform Lipschitzianity of the driver.
We now compute the so-called utility indifference strategy, that is, the deviation from the pure investment strategy for an investor who aims to maximize the expected utility of the terminal wealth in presence of the claim. 
Formally, the utility indifference strategy is a process $\overline \theta=\{\overline \theta_t,\ t \in \llbracket 0, T\wedge \tau\rrbracket\}$ defined by $\overline \theta_t:=\theta^*_t-\theta^0_t$, for every $t \in \llbracket 0, T\wedge \tau\rrbracket$.
We observe that
\begin{align}
\theta^*_t=\left(\frac{\mu^S(t)}{\alpha(\sigma^S(t))^2}+\frac{\widehat \gamma^1_t}{\alpha\sigma^S(t)};\frac{\mu^B(t, \mu_t)}{\alpha(c^B(t, \mu_t))^2}+\frac{\widehat \gamma^2_t}{\alpha c^B(t,\mu_t)}\right),\\
\theta^0_t=\left(\frac{\mu^S(t)}{\alpha(\sigma^S(t))^2}+\frac{\phi^1_t}{\alpha\sigma^S(t)};\frac{\mu^B(t, \mu_t)}{\alpha(c^B(t, \mu_t))^2}+\frac{\phi^2_t}{\alpha c^B(t, \mu_t)}\right),
\end{align}
for every $t \in \llbracket 0, T\wedge \tau\rrbracket$, and then, taking the difference we get \eqref{eq:hedg-strategy}.
\end{proof}


\subsection{The indifference price and the utility indifference strategy in the market model of Example \ref{C11}}

We now consider the framework of Example \ref{C11}, where the risky asset price process $S^1$ and the longevity bond price process $S^2$ are described by
\begin{align}
 \ud S_t^1 & = S_t^1\left(\mu^S(t, Y_t)\ud t +\sigma^S(t, Y_t) \ud W^1_t\right), \quad S_0^1 = s_0^1 \in \R^+,\\
\ud S^2_t & = S^2_t\left(\mu^B(t, Y_t)\ud t +d^B(t,Y_t)\ud W_t^3\right), \quad S_0^2=s_0^2\in \R^+
\end{align}
where the process $Y$ follows equation \eqref{def:Y}.
In this case the filtration $\widetilde \bF$ is the natural filtration of Brownian motions $W^1$ and $W^3$. We assume that the functions $\ds \frac{\mu^S(t,y)}{\sigma^S(t,y)}$ and $\ds \frac{\mu^B(t,y)}{d^B(t,y)}$ are bounded.

With abuse of notation we still denote by $\overline \P$ the probability measure equivalent to $\P$ on $\widetilde \G_T$ that makes the processes $\overline W^1=\{\overline W_t^1,\ t \in [0,T]\}$, $\overline W^3=\{\overline W_t^3,\ t \in [0,T]\}$ given by
\begin{align}
\overline W_t^1 & := W_t^1 + \int_0^t \frac{\mu^S(u, Y_u)}{\sigma^S(u, Y_u)} \ud u, \\
\overline W_t^3 & := W_t^3 + \int_0^t \frac{\mu^B(u,Y_u)}{d^B(u,Y_u)} \ud u,
\end{align}
for every $t \in [0,T]$, $(\widetilde \bG,\overline \P)$-Brownian motions and the $(\widetilde \bG, \overline \P)$-mortality intensity remains $\pi(\lambda)(1-H)$. Again we have that $\overline W^1$, $\overline W^3$ are $(\widetilde \bF,\overline \P)$-Brownian motions and the restriction of  $\overline \P$ on $\widetilde \F_T$ represents the unique martingale measure on the complete primary financial-insurance market.

The next result provides a representation for the indifference price of the pure endowment contract.
\begin{proposition}
 The indifference price is given by $p_t^\alpha = \overline p_t^\alpha \I_{\{t < \tau \}}$,  where $(\overline p^\alpha, \overline \gamma^1, \overline \gamma^3) $  is the unique solution to the following BSDE
 \begin{equation}
 \overline p_t^\alpha  = \xi - \int_t^T \overline \gamma_s^1 \ud \overline W_s^1 - \int_t^T \overline \gamma_s^3 \ud \overline W_s^3 + \frac{1}{\alpha}\int_t^T \left(e^{-U_s^0}e^{-\alpha \overline p_s^\alpha} - 1\right)\widehat \pi_s(\lambda) \ud s, \label{eq:pbar}
 \end{equation}
with $\overline p^\alpha$ bounded and $ \overline \gamma^i$ are $\widetilde \bF$-predictable and such that $\mathbb{E}^{\overline \P}\left[\int_0^T|\gamma^i_s|^2 \ud s\right]<\infty$ for  $i=1,3$, where
\begin{equation*}
 U^0_t =  - \frac{1}{2} \mathbb{E}^{\overline \P} \left[ \int_t^T\left[ \left(\frac{\mu^S(s, Y_s)}{\sigma^S(s, Y_s)}\right)^2 +\left(\frac{\mu^B(s,Y_s)}{c^B(s,Y_s)}\right)^2 \right]\ud s\bigg{|}\widetilde \F_t\right].
\end{equation*}

Moreover, for every $t \leq T\wedge \tau$,  the utility indifference hedging strategy is given by
 \begin{align}
\overline\theta_t=\left(\frac{\overline \gamma^1_t}{\sigma^S(t, Y_t)};\frac{\overline \gamma^3_t}{d^B(t, Y_t)}\right). \label{eq:strategy}
\end{align}

 \end{proposition}

The proof follows the same lines as that of Proposition \ref{Bsde_ind}.

\section{Concluding remarks}\label{sec:conclusion}

In this paper we discuss the indifference price of
a pure endowment contract in a market model where the tradable assets are given by a money market account, a stock and a longevity bond. Our modeling setting has three peculiar characteristics. First, we include mutual dependence between the financial and the insurance frameworks. Second, the mortality intensity of the individual is assumed to be different from the mortality intensity of the reference population. Indeed, due to differences in socioeconomic profiles, the hazard rates of the population typically differ from those of the policyholders. This implies in particular that even if the primary market is complete, when including the insurance derivative incompleteness arises. Third, the mortality intensity of the individual is not directly observable by the insurance company, which, at any time only knows if the individual is still alive or not. We define the indifference price of the insurance contract in terms of two stochastic control problems, with and without the insurance liability, which are solved using a BSDE approach. The optimization problem without the derivative leads to a continuous BSDE with quadratic driver for which existence and uniqueness follow by classical results. Therefore we mainly concentrate on the investment problem with the insurance derivative. This problem is formulated under partial information, and hence it first required to apply filtering techniques to derive an equivalent stochastic control problem with respect to the observable filtration only. The value process is then characterized in terms of a BSDE with a jump for which we prove existence and uniqueness of the solution up to $\tau\wedge T$ when the market prices of risks are bounded. Considering the BSDE up to time $T \wedge \tau$ is possible since the value processes of the optimization problems with and without the insurance liability coincide for all $T>t>\tau$, and this allows us to reduce the equation to a continuous BSDE with quadratic-exponential driver for which we provide existence and uniqueness of a bounded solution.
We also discuss two examples where we have a more explicit representation of the indifference price process of the pure endowment contract.
The first example corresponds to the case where the primary market is unaffected by additional observable economic and environmental factors, while the second one describes the case where the mortality intensity is a function of these additional factors modeled through a process $Y$. In both cases the primary market is complete.  However, incompleteness still arises when introducing the insurance derivative due to additional uncertainty from the policyholder mortality. In these two particular settings the indifference price process is characterized as the solution of a BSDE with respect to the market filtration $\widetilde \bF$ under a specific martingale measure $\overline \P$ equivalent to the physical measure $\P$. A similar representation cannot be obtained in the general case due to incompleteness of the primary financial-insurance market.

\begin{center}
{\bf Acknowledgements}
\end{center}
The authors are thankful to an anonymous referee for valuable suggestions that helped to improve the paper.
The authors are members of the Gruppo Nazionale per l'Analisi Matematica, la Probabilit\`a e le loro Applicazioni (GNAMPA) of the Istituto Nazionale di Alta Matematica (INdAM) and the work was partially supported by the GNAMPA project number 2017/0000327.  Part of this paper was written while the second-named author was affiliated with the School of Mathematics, University of Leeds, UK. The third-named author  was supported by {\em Universit\`a degli Studi di Perugia - Fondo ricerca di base  Esercizio 2015 - Project: Il problema della copertura di titoli derivati soggetti a rischio di credito in informazione parziale}.

\appendix

\section{Filtering}\label{appendix:filtering}
The goal of this section is to prove Proposition \ref{nuova_filter}. To this, we first need some preliminary results.

We recall here that
\begin{gather}
\widetilde{\bF}= \bF^{W^1} \vee \bF^{W^2} \vee \bF^{W^3}, \quad   \bF = \widetilde{\bF} \vee \bF^{Z}\\
\quad \widetilde{\bG} = \widetilde{\bF}  \vee \bF^H,  \quad \bG = \bF \vee \bF^H,
\end{gather}
where $W^j=\{W_t^j,\ t \in [0,T]\}$, $j=1,2,3$,  are $\P$-independent Brownian motions and $\P$-independent of $Z$. The process $H$ is the death indicator given by $H_t := \I_{\{\tau \leq t\}}$, for every $t \in [0,T]$, with $( \bG, \P)$-predictable intensity given by
$ \{ (1-H_{t^-}) \lambda(t, \mu_t, Z_{t^-}), t \in [0,T]\}.$
First, we derive the dynamics of the filter $\pi=\{\pi_t, \ t \in [0,T]\}$ which provides the conditional distribution of the unobservable process $Z$, given the observation flow $\widetilde\bG$. In other terms, we will compute
\begin{align*}
\pi_t(f)=\esp{f(Z_t)\Big{|} \widetilde \G_t}.
\end{align*}
for every $t \in [0,T]$ and for every $f \in \mathcal D$. This is essential to compute the $(\widetilde \bG, \P)$-predictable intensity of $H$, given by $ \{(1-H_{t^-}) \pi_{t^-}(\lambda), \  t \in [0,T]\},$
where  $\pi_{t}(\lambda)$ indicates $\pi_{t}(\lambda(t,\mu_t, \cdot) )$ (i.e. $ \pi_{t}(\lambda)(\omega)  = \esp{ \lambda(t,\mu_t(\omega), Z_t)  \Big{|} \widetilde\G_t}(\omega)$, $\forall \omega \in \Omega$).\\
The following result characterizes the filter as the unique strong solution of the so-called Kushner-Stratonovich equation.
\begin{proposition} \label{prop:KS}
If Assumption \ref{mgp} holds, the function $\lambda(t,\mu,z)>0 $ is continuous in $z \in \mathcal Z$
and $\sup_{(t,\mu,z)\in [0,T]\times \R^+ \times \mathcal Z} \lambda(t,\mu,z) < \infty$,  then the filter $\pi = \{\pi_t,\ t \in [0,T]\}$ is the unique strong solution to the equation
\begin{equation} \label{KS}
\pi_t(f) = f(z_0) + \int_0^t \pi_s(\L^Z f) \ud s + \int_0^t \frac{\pi_{s^-}(\lambda f) - \pi_{s^-}(\lambda)\pi_{s^-}(f) }{\pi_{s^-}(\lambda)} \ud M_s^\tau,
\end{equation}
for every $t \in [0,T]$ and for all  $f \in \mathcal{ D} $.
\end{proposition}

\begin{proof}
To prove the result we use the Innovation approach. Since $W^1, W^2$ and $W^3$ are $(\widetilde \bG, \P)$-Brownian motions and  the process $M^\tau$ given in \eqref{def:mtau} is a $(\widetilde \bG, \P)$-jump martingale, we define the Innovation process by $(W^1, W^2, W^3, M^\tau)$. \\
For every function $f \in \mathcal D$, by projecting equation \eqref{eq:f_semimg} on $\widetilde \bG$, we get
\begin{align}\label{eq:pi_semimg}
\pi_t(f)=\esp{f(Z_t)| \widetilde \G_t}= f(z_0) + \int_0^t \pi_s({\L}^Z f) \ud s + M^{(1)}_t, \quad t \in [0,T],
\end{align}
where $M^{(1)}:=\{M^{(1)}_t, \ t \in [0,T]\}$ is the $(\widetilde \bG, \P)$-martingale given by $M^{(1)}_t := \esp{M^Z_t\Big{|} \widetilde \G_t} + \esp{\int_0^t {\L}^Z f(Z_s) \ud s  \Big{|} \widetilde \G_t} - \int_0^t  \pi_s({\L}^Z f) \ud s $ (see, e.g. \citet[Chapter IV, Theorem T1]{bremaud1981}).  By the Martingale Representation Theorem (see, e.g. \citet[Theorem 3.34]{lipster2001statistics}) with respect to filtration $\widetilde \bG$ and probability measure $\P$, there exist $\widetilde \bG$-adapted processes $\widehat h^i=\{\widehat h^i_t, \ t \in [0,T]\}$ and a $\widetilde \bG$-predictable process $\widehat\varphi=\{\widehat\varphi_t, \  t \in [0,T]\}$ satisfying
\begin{align}\label{eq:integrability_cond}
\esp{\int_0^T \left(\sum_{i=1}^3 (\widehat h^i_t)^2 + |\widehat \varphi_t| \lambda(t, \mu_t, Z_t)\right) \ud t} < \infty.
\end{align}
and  such that
\begin{align*}
M^{(1)}_t= M^{(1)}_0+ \sum_{i=1}^3\int_0^t \widehat h^i_s \ud W^i_s + \int_0^t \widehat\varphi_s \ud M^\tau_s, \quad t \in [0,T].
\end{align*}
In order to identify the processes $\widehat h^i$ for $i=1,2,3$ and $\widehat \varphi$, we observe that $\esp{f(Z_t) W^i_t\Big{|} \widetilde \G_t}=\esp{f(Z_t)\Big{|} \widetilde \G_t}W^i_t$; then, by computing both quantities and comparing the finite variation parts we get that $\widehat h^i_t=0$ $\P$-a.s. for every $t \in [0,T]$. Moreover, it holds that for every process $U=\{U_t, \ t \in [0,T]\}$ of the form $U_t=\int_0^t C_s \ud H_s$ for some $(\widetilde \bG, \P)$-predictable process $C=\{C_t, \ t \in [0,T]\}$,   $\esp{f(Z_t) U_t\Big{|}\widetilde \G_t}=\esp{f(Z_t)\Big{|}\widetilde \G_t} U_t$. Then, by computing separately the right-hand side and the left-hand side of the equation and comparing the finite variation parts we get that
\begin{align*}
C_t \widehat\varphi_t (1-H_{t^-}) \pi_{t^-}(\lambda) = C_t \left( \pi_{t^-}(f \lambda)-  \pi_{t^-}(f) \pi_{t^-}(\lambda) \right)  (1-H_{t^-}), \ t \in [0,T],
\end{align*}
and since the process $C$ is arbitrary, we obtain that on the set $\{t\leq \tau\}$
\begin{align}\label{eq:hat_gamma}
 \widehat\varphi_t  = \frac{\pi_{t^-}(f \lambda)}{\pi_{t^-}(\lambda)} -  \pi_{t^-}(f), \ \P-\mbox{a.s.}.
\end{align}
Therefore, $M^{(1)}_t= M^{(1)}_0+\int_0^t \left( \frac{\pi_{s^-}(f \lambda)}{\pi_{s^-}(\lambda)} -  \pi_{s^-}(f) \right) \ud M^\tau_s$, and plugging this expression in \eqref{eq:pi_semimg} we get the result.\\
Uniqueness can be proved  as in  \citet[Theorem 3.3]{cecicolaneri2012},  by applying the Filtered Martingale Problem approach. We start by observing that for any $f \in \mathcal{ D}$ and any measurable function $\phi$ on $\{0,1\}$, we have
\begin{align*}f(Z_t) \phi(H_t) = &f(Z_0) \phi(H_0)   + M^{f,\phi}_t\\
&+ \int_0^t \{ \L^Z f(Z_s) +  [\phi( H_{s^-} + 1) - \phi( H_{s^-})] (1 -  H_{s^-}) \lambda(s, \mu_s, Z_s) \} \ud s, \end{align*}
for every $t \in [0,T]$, where  $M^{f,\phi} = \{M^{f,\phi}_t,\ t \in [0,T]\}$ is a $(\widetilde \bG, \P)$-martingale. Then, for any $\mu \in \R^+$ it follows that the pair $(Z, H)$ solves the martingale problem for the operator $\mathcal L^\mu$ defined by
\begin{align*}\L^\mu \psi(t,z, h) := \frac{\partial \psi}{\partial t} (t,z,h)
+ {\L}^Z \psi(t,z,h) +[ \psi(t,x,h + 1)-\psi(t,z,h)](1-h) \lambda(t,\mu,z)\end{align*}
for every function  $\psi$ in the domain $\mathcal{D}^\mu$ of $\mathcal L^\mu$, where $\mathcal{D}^\mu$ consists of all bounded functions  $\psi(t,x,h)$ having continuous partial derivatives with respect to $t$ and such that $\psi(t, \cdot, h) \in \mathcal{D}$, $\forall (t,h) \in [0,T] \times \{0,1\}$.
The pair $(\L^\mu, \mathcal{D}^\mu)$ satisfies Assumption \ref{mgp}, where we replace $\mathcal Z$ with  $[0,T] \times \mathcal Z \times \{0,1\}$. By \citet[Theorem 3.3]{kurtz1988unique} we get that the Filtered Martingale Problem for the operator $\L^\mu$ is well posed. Then, we can apply  \citet[Theorem 3.3]{cecicolaneri2012}, which ensures strong uniqueness.

\end{proof}
We make a few remarks. Over the set $\{t <\tau < T\}$ the filter solves a nonlinear equation given by
\begin{align}\label{eq:1n}
\pi_t(f) = f(z_0) + \int_0^t  \left( \pi_s(\L^Z f) \ud s   - \pi_{s}(\lambda f) +  \pi_{s}(\lambda)\pi_{s}(f)  \right) \ud s,
\end{align}
and at time $\tau<T$ we get that
\begin{align}\label{taun}
\pi_{\tau}(f)=  \pi_{\tau^-} (f) + \frac{\pi_{\tau^-}(\lambda f) - \pi_{\tau^-}(\lambda)\pi_{\tau^-}(f) }{\pi_{\tau^-}(\lambda)} =   \frac{\pi_{\tau^-}(\lambda f)}{\pi_{\tau^-}(\lambda)} .
\end{align}
Finally, after the jump, that is over the set  $\{\tau<t\le T\}$,  the filtering equation is linear, of the form
\begin{align}\label{eq:2n}
\pi_t(f) = \pi_{\tau}(f)+ \int_\tau^t  \pi_s(\L^Z f) \ud s.
\end{align}
In order to obtain an explicit expression for the filter, we apply a suitable change of probability measure, which allows to obtain a linear equation for the unnormalized filter, known in literature as  the Zakai equation. To this aim we introduce the process $L=\{L_t, \ t \in [0,T]\}$ by
\begin{align}\label{L}
L_t := \E \left ( \int_0^\cdot \frac{1 - \lambda(s, \mu_s, Z_{s^-}) }{ \lambda(s, \mu_s, Z_{s^-})  }  \{ \ud H_s - (1-H_{s^-}) \lambda(s, \mu_s, Z_{s^-})   \ud s\} \right )_t,
\end{align}
for every $t \in [0,T]$, where $\E $ denotes the Dol\'eans-Dade exponential. We assume that $L$ is a $( \bG, \P)$-martingale. This is implied, for instance, by the condition
\begin{equation} \label{Nov}\esp{ e^{ \int_0^T \frac{(1 - \lambda(s, \mu_s, Z_s)  )^2 }{ \lambda(s, \mu_s, Z_s)  } (1 - H_s)  \ud s}} < \infty, \end{equation}
and satisfied, in particular,  if the function $\lambda(t,\mu,z)$ is bounded from below and above. Then we define the probability measure  $\Q$ equivalent to $\P$ by
$$
\frac{\ud \Q}{\ud \P}\Big|_{\G_t} := L_t,
$$
for every $t \in [0,T]$. By the Girsanov Theorem we have that
\begin{align*}
\left\{H_t - \int_0^t (1 -H_{s^-}) \ud s, \quad t \in [0,T]\right\}
\end{align*}
is a $(\bG, \Q)$-martingale, and  the process $\{1 -H_{t^-}, \ t \in [0,T]\}$ provides the $(\bG, \Q)$-predictable intensity of $H$.
We introduce the unnormalized filter, which is the finite measure valued  process $\rho = \{\rho_t,\ t \in [0,T]\}$ given by
$$
\rho_{t}(f) := \mathbb E^\Q \left[ L^{-1}_t f(Z_t)  \Big{|}  \tilde \G_t \right], \quad t \in [0,T],
$$
for every bounded measurable function $f$. By applying the Kallianpur-Striebel formula we get that
\begin{equation} \label{KSt}
\pi_t(f)= \esp{f(Z_t)\Big{|}\widetilde \G_t} =
 \frac{\rho_t (f)}{\rho_t (1)}, \quad t \in [0,T],
\end{equation}
for every bounded and measurable function $f(z)$, where $ \rho_t (1) := \mathbb E^\Q \left[ L^{-1}_t  \Big{|}  \tilde \G_t  \right]$. The dynamics of process $\rho(1)$ can be easily computed by observing that the $(\widetilde \bG, \P)$-intensity of $H$ is given by $\{(1-H_{t^-}) \pi_{t^-}(\lambda), t \in [0,T]\}$ and the $(\widetilde \bG, \Q)$-intensity of $H$ is $\{1 -H_{t^-}, \ t \in [0,T]\}$, then we get that $\rho(1)$ is an exponential martingale satisfying the following stochastic differential equation
\begin{equation}
\ud \rho_t (1)  = \rho_{t ^-}(1)  ( \pi_{t^-}(\lambda) - 1)  (\ud H_t - (1-H_{t^-}) \ud t),  \quad \rho_0 (1)=1.
\end{equation}
Then,  by applying the product rule  to $\rho_t(f) = \pi_t(f) \rho_t(1)$ and using equation \eqref{KS},  we get that
\begin{equation} \label{Z}
\rho_t(f) = f(z_0) + \int_0^t \rho_s(\L^Z f) \ud s + \int_0^t \rho_{s^-}( f(\lambda-1)) [ \ud H_s - (1-H_{s^-}) \ud s],
\end{equation}
for every $t \in [0,T]$, where $\rho_t(\lambda)$ indicates $\rho_t(\lambda(t,\mu_t, \cdot) )$.  Over the set $\{t < \tau<T\}$ this equation reduces to
\begin{equation}\label{Z1}
\rho_t(f) = f(z_0) + \int_0^t (\rho_s(\L^Z f) - \rho_{s}( f\lambda) + \rho_{s}( f ) )\ud s
\end{equation}
and the solution can be computed explicitly, as shown in the following proposition.

\begin{proposition}\label{RHO}
Let
\begin{equation} \label{formula1}
\tilde \rho_t(f)(\omega) :=  \esp{ f(Z_t) e^{-\int_0^t ( \lambda (u,\mu_u(\omega), Z_u) - 1 ) \ud u} }, \quad t \in [0,\tau(\omega)).
\end{equation}
Then,  $\tilde \rho$ solves equation \eqref{Z1} over $\{ t < \tau\}$.   \end{proposition}

 \begin{proof}
For any fixed trajectory $t \to \mu_t(\omega)$ of process $\mu$, we set $\gamma_t := e^{-\int_0^t ( \lambda (s,\mu_s(\omega), Z_s) - 1 ) \ud s}$. By the product rule we get
$$
\ud ( f(Z_t) \gamma_t)   =  \gamma_t  [ \L^Z f(Z_t) - f(Z_t) (\lambda (t,\mu_t(\omega), Z_t) - 1 ) ] \ud t  + \gamma_t \ud M^Z_t.
$$
Now, taking expectation
$$
\esp{ f(Z_t) \gamma_t } = f(z_0) + \int_0^t \esp{ \gamma_s [ \L^Z f(Z_s) - f(Z_s) (\lambda (s,\mu_s(\omega), Z_s) - 1 ) ]  }\ud t.
$$
Then,  we get that $\esp{ f(Z_t) e^{-\int_0^t ( \lambda (s,\mu_s(\omega), Z_s) - 1 ) \ud s} }$ solves equation \eqref{Z1} for any fixed trajectory of the process $\mu$ and this concludes the proof.
\end{proof}

Finally, we give the proof of Proposition \ref{nuova_filter}. Let us observe that we do not require the assumption that $L$,
given in \eqref{L}, to be a $( \bG, \P)$-martingale.

 \begin{proof}[Proof of Proposition \ref{nuova_filter}]
Let, for $t <\tau<T$,
$$\widetilde \pi_t(f):=  \frac{\widetilde \rho_t (f)}{\widetilde \rho_t (1)},$$
where $\widetilde \rho$ given in \eqref{formula1}, for $t=\tau<T$
$$\widetilde \pi_{\tau}(f):=  \frac{\widetilde \pi_{\tau^-}(\lambda f)}{ \widetilde \pi_{\tau^-}(\lambda)},$$
and for $ \tau<t\leq T$,
$$ \widetilde \pi_t(f) :=  {\mathbb E}_{\tau,  \widetilde \pi_\tau}[ f(Z_t) ],$$
where $ {\mathbb E}_{\tau,\widetilde \pi_\tau}$ denotes the conditional expectation given the law of $Z$ at time $\tau$ equals to $\widetilde\pi_\tau$.
Then,  by a direct computation we can show that $\widetilde \pi$ solves equation \eqref{KS} and
by strong uniqueness  (see Proposition \ref{prop:KS}) we get that $\widetilde \pi = \pi$.

In particular, defining the $\widetilde \bF$-adapted process  $\widehat  \pi:=\{\widehat \pi_t, t \in [0,T]\}$ given by \eqref{eq:hat_pi},
we get that on $\{t < \tau\}$ the filter $\pi$ coincide with the process $\widehat \pi$.

\end{proof}

\section{Longevity bond price }\label{app:longevity}
We start from a filtered probability space $(\Omega, \widetilde \F, \widetilde \bF, \Q)$, where $\Q$ is a risk neutral measure equivalent to $\P$.
The objective of this section is to characterize the fair price of the longevity bond under the measure $\Q$ and get the $\P$-price dynamics via change of measure.
Let $W^{1,\Q}=\{W^{1,\Q},\ t \in [0,T]\}$, $W^{2,\Q}=\{W^{2,\Q},\ t \in [0,T]\}$, $W^{3,\Q}=\{W^{3,\Q},\ t \in [0,T]\}$ be $\Q$-independent Brownian motions and define  the density process $L^\P=\{L^\P_t, \ t \in [0,T]\}$   of $\P$ with respect to $\Q$ by
\begin{equation} \label{def:q}
L^\P_t\!:=\!\left.\frac{\ud \P}{\ud \Q}\right|_{\widetilde \F_t}\!\! \!\!= \E\! \left(\int_0^. \!\!\frac{\mu^S(u, Y_u)}{\sigma^S(u, Y_u)} \ud W_u^{1,\Q} - \!\!\int_0^. \!\!\alpha^{\mu}(u,\mu_u,Y_u) \ud W_u^{2,\Q} -\!\! \int_0^. \!\!\alpha^{Y}(u,\mu_u,Y_u) \ud W_u^{3,\Q}\right)_t,
\end{equation}
for every $t \in [0,T]$, where functions $\mu^S(t,y)$, $\sigma^S(t,y)$, $\alpha^\mu(t,\mu,y)$ and $\alpha^Y(t,\mu,y)$ are measurable and  such that $L^\P$ is an $(\widetilde \bF,\Q)$-martingale. By applying Girsanov theorem we get that processes $W^{1}=\{W^{1},\ t \in [0,T]\}$, $W^{2}=\{W^{2},\ t \in [0,T]\}$, $W^{3}=\{W^{3},\ t \in [0,T]\}$ respectively defined by
\begin{align}
W_t^{1} & := W_t^{1,\Q} - \int_0^t \frac{\mu^S(u, Y_u)}{\sigma^S(u, Y_u)} \ud u, \quad t \in [0,T],\\
W_t^{2} & := W_t^{2,\Q} + \int_0^t \alpha^{\mu}(u,\mu_u,Y_u) \ud u, \quad t \in [0,T],\\
W_t^{3} & := W_t^{3,\Q} + \int_0^t \alpha^{Y}(u,\mu_u,Y_u) \ud u, \quad t \in [0,T],
\end{align}
are $\P$-independent $\widetilde \bF$-Brownian motions. Following \citet{cairns2006pricing}, a longevity bond is defined  as a zero-coupon bond that pays out the value of the survivor or longevity index at time $T$. Then, its discounted price process $S^2$ at any time $t$ is given by
\begin{align}
S^2_t& = \espq{S_T^\mu\Big{|}\widetilde \F_t}=\espq{\exp\left(-\int_0^T \mu_s \ud s\right)\Bigg{|}\widetilde \F_t}=e^{-\int_0^t \mu_s \ud s}\espq{\exp\left(-\int_t^T \mu_s \ud s\right)\bigg{|} \widetilde \F_t}, \label{def:lbp}
\end{align}
for every $ t \in [0,T]$.
We write the dynamics of the pair $(\mu, Y)$ with respect to $\Q$ as
\begin{align}
\ud \mu_t & = ( b^\mu(t,\mu_t,Y_t) + \alpha^\mu(t,\mu_t,Y_t) ) \ud t + \sigma^\mu(t,\mu_t,Y_t) \ud W_t^{2,\Q},\quad \mu_0 \in \R^+, \\
\ud Y_t& = ( b^Y(t,Y_t) + \alpha^Y(t,\mu_t,Y_t) ) \ud t + \sigma^Y(t,Y_t) \ud W_t^{3,\Q}, \quad Y_0=y_0 \in \R.
\end{align}
Since the pair $(\mu, Y)$ is an $(\widetilde \bF, \Q)$-Markov process with infinitesimal generator  $\mathcal L^{\mu,Y}$ under $\Q$, setting
\begin{equation} \label{F}
F(t,\mu,y):=\espq{\exp\left(-\int_t^T \mu_s \ud s\right)\bigg{|}\mu_t=\mu,\ Y_t=y},
\end{equation}
we get that relation \eqref{def:lbp} can be written as
\begin{equation}\label{eq:s2f}
S_t^2=e^{-\int_0^t \mu_s \ud s}F(t,\mu_t,Y_t), \quad t \in [0,T].
\end{equation}
The function $F(t,\mu,y)$ is strictly positive and bounded from above by $1$.
If the function $F(t,\mu,y)$ is sufficiently regular, that is, $ F\in \mathcal C_b^{1,2,2}([0,T] \times \R^+ \times \R)$, then it can be characterized via the Feynman-Kac formula as the solution of
the boundary problem
\begin{equation}\label{FK}
\left\{
\begin{array}{ll}
\ds \frac{\partial F}{\partial t}(t,\mu,y) + \mathcal L^{\mu,Y}F(t,\mu,y) -\mu F(t,\mu,y) = 0, \quad (t,\mu,y) \in [0,T) \times \R^+ \times \R,\\
F(T,\mu,y)=1, \quad (\mu,y) \in \R^+ \times \R.
\end{array}
\right.
\end{equation}
To ensure that Feynman-Kac formula \eqref{F} applies, we make the following set of assumptions.
\begin{assumption}\label{ass:coeff_ass}
Functions $b^\mu(t,\mu,y),\ b^Y(t,y)$, $\alpha^\mu(t,\mu,y),\alpha^Y(t,y)$, $\sigma^\mu(t,\mu,y),\ \sigma^Y(t,y)$ are continuous in all variables and satisfy sublinear growth-conditions on $(\mu, y) \in \R^+ \times \R$, uniformly in $t \in [0,T]$. Moreover, $b^\mu(t,\mu,y),\ b^Y(t,y)$,  $\alpha^\mu(t,\mu,y),\alpha^Y(t,y)$, and $(\sigma^\mu(t,\mu,y))^2,\ (\sigma^Y(t,y))^2$ are Lipschitz continuous on $(\mu, y) \in \R^+ \times \R$, uniformly in $t \in [0,T]$, and $\sigma^\mu(t,\mu,y),\ \sigma^Y(t,y)$ are bounded from below.
\end{assumption}
The following result shows existence and uniqueness for the solution of the boundary problem \eqref{FK}.
\begin{proposition}
Under Assumption \ref{ass:coeff_ass},
there exists a unique classical solution $F$ to the boundary problem \eqref{FK} and the Feynman-Kac representation \eqref{F} holds.
\end{proposition}

\begin{proof}
The result follows from  \citet[Theorem 1]{heath2000martingales}.
Indeed, condition $(A2)$ in \citet[Theorem 1]{heath2000martingales} is a consequence of sublinear growth-condition  and Lipschitz continuity of the coefficients and by \citet[Lemma 2]{heath2000martingales}, $F(t,\mu,y)$ given in \eqref{F} is continuous in $[0,T] \times \R^+ \times \R$.
\end{proof}

Now, we can apply It\^o's formula to $S^2$ given in \eqref{eq:s2f} and since $F$ is solution of \eqref{FK},we get that  $S^2$ solves
\begin{equation}
\ud S_t^2  = S_t^2 \left(c^B(t,\mu_t,Y_t)\ud W_t^{2,\Q} + d^B(t,\mu_t,Y_t)\ud W_t^{3,\Q}\right), \quad S_0^2=s_0^2\in \R^+,\end{equation}
where we have set
\begin{equation}\label{def:cb-db}
c^B(t,\mu,y):=\ds \frac{\sigma^\mu(t,\mu,y)}{F(t,\mu,y)}\frac{\partial F}{\partial \mu}(t,\mu,y), \qquad d^B(t,\mu,y):=\ds\frac{\sigma^Y(t,y)}{F(t,\mu,y)}\frac{\partial F}{\partial y}(t,\mu,y),
\end{equation}
for every $(t,\mu,y) \in [0,T] \times \R^+ \times \R$. Finally, the $\P$-dynamics of the process $S^2$ is given by \eqref{eq:long_bond} where we have set $\mu^B(t, \mu_t, Y_t):=c^B(t,\mu_t,Y_t)\alpha^\mu(t,\mu_t,Y_t)+d^B(t,\mu_t,Y_t) \alpha^Y(t,\mu_t,Y_t)$.

\section{A Technical result}\label{appendix:tech_res}

\begin{lemma}\label{lemma:density}
In the modeling framework outlined in Section \ref{sec:model}, the so-called {\em density hypothesis} holds with respect to filtration $\widetilde\bF$. Precisely,
for every $t \in [0,T]$, there exists a function $\widetilde \beta(t, \cdot):\R^+ \to \R^+$, such that $(t,u) \mapsto \widetilde \beta(t,u)$ is $\widetilde \F_t \otimes \mathcal{B}(0,\infty)$-measurable and such that
\[
\P(\tau>s|\widetilde\F_t) = \int_s^\infty \widetilde \beta(t,u) \ud u, \quad s \in \R^+,
\]
and $\widetilde \beta(t,u)\I_{\{t\geq u\}}=\widetilde \beta(u,u)\I_{\{t\geq u\}}$.
\end{lemma}

\begin{proof}
Firstly, note that the Cox construction for the random time $\tau$ describing the residual life time of the policyholder, ensures that the density hypothesis is fulfilled with respect to the filtration $\bF$, see \citet[Section 5]{jeanblanc2009progressive}.
That is, there exists a map $\beta(t, \cdot):\R^+ \to \R^+$, such that $(t,u) \mapsto \beta(t,u)$ is $\F_t \otimes \mathcal{B}(0,\infty)$-measurable and such that
\begin{equation} \label{beta}
\P(\tau>s|\F_t) = \int_s^\infty \beta(t,u) \ud u, \quad s \in \R^+,
\end{equation}
and $\beta(t,u)\I_{\{t\geq u\}}=\beta(u,u)\I_{\{t\geq u\}}$.
Precisely, in our setting $$\beta(t,u)=\esp{\lambda(u, \mu_u, Z_u) e^{-\int_0^u \lambda(r, \mu_r, Z_r)\ud r}|\F_t}.$$ Conditioning in  \eqref{beta} with respect to $\widetilde \F_t \subseteq \F_t$ and applying the Fubini theorem yield
  \begin{align*}
\P(\tau>s|\widetilde \F_t) & = \condespftilde{\int_s^\infty \beta(t,u) \ud u}\\
& = \int_s^\infty\condespftilde{\lambda(u, \mu_u, Z_u)e^{-\int_0^u \lambda(r, \mu_r, Z_r)\ud r}}\ud u\\
& = \int_s^\infty \widetilde\beta(t,u) \ud u,
  \end{align*}
where  $\widetilde \beta(t,u) := \condespftilde{\lambda(u, \mu_u, Z_u)e^{-\int_0^u \lambda(r, \mu_r, Z_r)\ud r}}$. Note that, since $\mu$ is $\widetilde \bF$-adapted and the process $Z$ is independent of $\widetilde \bF$, for any $t \ge u$
\begin{equation}
\widetilde \beta(u,u) (\omega) = \esp{\lambda(u, \mu_u(\omega), Z_u)e^{-\int_0^u \lambda(r, \mu_r(\omega), Z_r)\ud r}} =\widetilde \beta(t,u)(\omega), \quad \forall \omega \in \Omega.
\end{equation}
and this concludes the proof.
\end{proof}

\end{document}